\documentclass[preprint,11pt]{imsart}
\usepackage{amsmath,amssymb,amsthm}
\usepackage{mathrsfs}
\usepackage[OT1]{fontenc}
\usepackage[utf8]{inputenc}
\usepackage{arydshln}
\usepackage[colorlinks,citecolor=blue,urlcolor=blue]{hyperref}
\usepackage{color}
\usepackage{indentfirst}
\usepackage{multirow}
\usepackage{float}
\usepackage{subfig}
\usepackage[english]{babel}
\usepackage{bm}
\usepackage[dvips]{graphicx}
\usepackage{multicol}

\usepackage[round]{natbib}
\usepackage{a4wide}

\startlocaldefs
\numberwithin{equation}{section} \theoremstyle{plain}
\newtheorem{theorem}{Theorem}[section]
\newtheorem{lemma}{Lemma}[section]
\newtheorem{corollary}{Corollary}[section]
\newtheorem{proposition}{Proposition}[section]

\endlocaldefs

\begin{document}

\newcommand\cov{\mathop{Cov}}
\newcommand\tr{\mathop{tr}}
\newcommand\E{\mathbb{E}}
\newcommand\N{\mathcal{N}}
\newcommand\lb{\left(}
\newcommand\rb{\right)}
\newcommand\veps{\varepsilon}

\renewcommand{\d}[1]{\ensuremath{\operatorname{d}\!{#1}}}

\newcommand\re[1]{{\color{red}#1}}  
\newcommand\bl[1]{{\color{blue}#1}}
\newcommand\gr[1]{{\color{green}#1}}

\begin{frontmatter}
  \title{Testing the sphericity of a covariance matrix when the dimension is much larger than the sample size}
  \runtitle{Sphericity Test}
%
%
%
%
  \begin{aug}
    \author{\fnms{Zeng} \snm{Li}\ead[label=e2]{u3001205@hku.hk}}
\and
    \author{\fnms{Jian-feng} \snm{Yao}\ead[label=e3]{jeffyao@hku.hk}}

    \runauthor{Z. Li and J. Yao}

   \affiliation{The University of Hong Kong}

    \address{Zeng Li, Jianfeng Yao\\
    Department of Statistics and Actuarial Science\\
      The University of Hong Kong\\
        \printead{e2,e3}
    }
  \end{aug}

  \begin{abstract}
  This paper focuses on the prominent sphericity test when the dimension $p$ is much lager than sample size $n$. The classical likelihood ratio test(LRT) is no longer applicable when $p\gg n$. Therefore a Quasi-LRT is proposed and its asymptotic distribution of the test statistic under both the null \re{and the alternative hypothesis} when $p/n\rightarrow\infty, n\rightarrow\infty$ is well established in this paper. We also re-examine the well-known John's invariant test for sphericity in this ultra-dimensional setting. An amazing result from the paper states that John's test statistic has exactly the same limiting distribution under the ultra-dimensional setting with under other high-dimensional settings known in the literature. Therefore, John's test has been found to possess the powerful {\it dimension-proof} property, which keeps exactly the same limiting distribution under the null with any $(n,p)$-asymptotic, i.e. $p/n\rightarrow[0,\infty]$, $n\rightarrow\infty$. All asymptotic results are derived for general population with finite fourth order moment. Numerical experiments are implemented to illustrate the finite sample performance of the results.
  \end{abstract}


\end{frontmatter}

\section{Introduction}
High dimensional data with dimension $p$ of same scale with or even larger than the number of observations $n$ has  applausive statistical applications in biology and finance recently. In particular, practical needs for testing gene-wise independence in genomic studies have inspired a wide range of discussions regarding test of structures of the covariance matrix.

In this paper, we consider the prominent sphericity test when the dimension $p$ is much larger than the sample size $n$. Let $X=(X_1,X_2,\cdots,X_n)$ be a $p\times n$ data matrix with $n$ independent and identically distributed $p-$dimensional random vectors $\{X_i\}_{1\leq i\leq n}$ with covariance $\Sigma=Var(X_i)$. Our interest is to test
\begin{equation}\label{testeq}
  H_0: \Sigma=\sigma^2 I_p~ \mbox{ vs. } ~H_1:\Sigma\neq \sigma^2 I_p,
\end{equation}
where $\sigma^2$ is an unknown positive constant. Among traditional tests are the likelihood ratio test(LRT) and John's invariant test.

Consider first the LRT with test statistic(\citet{Anderson84})
\begin{equation}\label{LRTeq}
  -2\log L_n=-2\log \lb\frac{(l_1\cdots l_p)^{1/p}}{\frac{1}{p}(l_1+\cdots+l_p)}\rb^{\frac{pn}{2}}=n\log\lb\frac{\overline{l}^p}{\prod_{i=1}^p l_i}\rb,
\end{equation}
where $\{l_i\}_{1\leq i\leq p}$ are the eigenvalues of $p-$dimensional sample covariance matrix $\frac1n\sum_{i=1}^n X_iX_i'=\frac1nXX'$, $X=\lb X_1,\cdots, X_n\rb$. If we let $n\rightarrow \infty$ while keeping $p$ fixed, classics asymptotic theory indicates that under the null hypothesis and assuming the population is normal,
\[-2\log L_n\xrightarrow{d} \chi^2_{\frac{1}{2}p(p+1)-1},\]
the chi-square distribution is further refined by the Box-Bartlett correction. However, this $\chi^2-$convergence becomes slow when the dimension $p$ increases so that the LRT (and its Box-Bartlett correction) is seriously biased when the dimension-to-sample size ratio $p/n$ is not small enough.

\citet{Wang13} \re{made} bias correction to the traditional LRT test under the regime where both $p,n\rightarrow\infty$, $p/n\rightarrow c\in(0,1)$. They derived that when $X=\{x_{ij}\}_{{1\leq i\leq p}\atop{1\leq j\leq n}}$ with i.i.d entries satisfying $\mathbb{E}(x_{ij})=0$, $\mathbb{E}|x_{ij}|^2=1$, $\nu_4:=\mathbb{E}|x_{ij}|^4<\infty$, and under $H_0$,
\begin{equation}
  -\frac{2}{n}\log L_n+(p-n)\log (1-\frac{p}{n})-p\xrightarrow{d} N\lb-\frac12 \log (1-c)+\frac{\nu_4-3}{2}c,-2\log(1-c)-2c\rb.
\end{equation}
Notice that here the scale parameter $\sigma^2$ in $H_0$ has been taken to be $\sigma^2=1$ as the LRT statistic is invariant under scaling. Extensive simulation study in \citet{Wang13} shows that this test is well adapted to high dimensions and has a very reasonable size and power for a wide range of dimension-sample size combinations $(p,n)$.
The LRT however requires that $p\leq n$ because when $p>n$, $n-p$ of the sample eigenvalues $\{l_i\}$ are null so that the likelihood ratio $L_n$ is identically null. In this paper, we introduce a quasi-LRT statistic which can be seen as a natural extension of the LRT statistic to the situation where $p>n$. The quasi-LRT test statistic is defined as
\begin{equation}\label{quasiLRT}
\mathcal{L}_n=\frac{p}{n}\log\frac{\lb\frac{1}{n}\sum_{i=1}^n\tilde{\lambda}_i\rb^n}{\prod_{i=1}^n\tilde{\lambda}_i},
\end{equation}
where $\{\tilde{\lambda}_i\}_{1\leq i\leq n}$ are eigenvalues of $n-$dimensional matrix $\frac{1}{p}X'X$. The main idea is that the companion matrix $X'X$ has exactly the same $n$ non-null eigenvalues with the sample covariance matrix $XX'$(up to some scaling). Therefore, the quasi-LRT test statistic removes all the null eigenvalues in the original LRT test statistic and we find that under the so-called ultra-dimensional asymptotic $p\gg n$, that is $p/n\rightarrow \infty$ and $n\rightarrow\infty$,
  \[\mathcal{L}_n-\frac{n}{2}-\frac{n^2}{6p}-\frac{\nu_4-2}{2} \xrightarrow{d} N\lb 0,1\rb.\]
Based on this asymptotic result, a quasi-LRT test can be conducted to test sphericity to compensate for the inapplicability of the traditional LRT in the ultra-dimension setting.

Next we consider John's invariant test for sphericity. \citet{John71, John72} studied the problem for normal populations and proposed the testing statistic
\begin{equation}\label{John}
 U=\cfrac{1}{p}tr\left[\lb\cfrac{\Sigma}{(1/p)tr(\Sigma)}-I_p\rb^2\right]=\frac{p^{-1}\sum_{i=1}^p(l_i-\overline{l})^2}{\overline{l}^2},
\end{equation}
\noindent
where $\overline{l}=\frac{1}{p}\sum_{i=1}^p l_i$. It has been proved that, as $n\rightarrow \infty$ while $p$ remain fixed, the limiting distribution of $U$ under $H_0$ is
\[nU-p\xrightarrow{d} \cfrac{2}{p} \chi^2_{p(p+1)/2-1}-p.\]
Contrary to the LRT, it has been noticed for a while that John's test does not suffer from high dimensions and this $\chi^2$ limit is quite accurate even when the ratio $p/n$ is not small.
\citet{Ledoit02} studied the $(n,p)$-consistency of this test statistic under normality assumptions. They proved that, when $n,p\rightarrow\infty$, $\lim_{n\rightarrow \infty} p/n\rightarrow c\in(0,+\infty)$,
\begin{equation}
  nU-p\xrightarrow{d} N(1,4).
\end{equation}
Meanwhile, when $p\rightarrow \infty$,
\[\cfrac{2}{p} \chi^2_{p(p+1)/2-1}-p\xrightarrow{d} N(1,4).\]
In other words, \citet{Ledoit02} extended the classical $n$-asymptotic theory (where $p$ is fixed) to the high-dimensional case where $p$ goes to infinity proportionally with $n$. Meanwhile, the robustness of John's test is explained in this proportional high-dimensional scheme.

\citet{Wang13} further relaxed the normality restriction and proved that, if $\{x_{ij}\}$ are i.i.d. with $\mathbb{E}x_{ij}=0$, $\mathbb{E}|x_{ij}|^2=1$, $\nu_4\triangleq \mathbb{E}|x_{ij}|^4<\infty$, then when $n,p\rightarrow\infty$, $\lim_{n\rightarrow \infty} p/n\rightarrow c\in(0,+\infty)$,
\begin{equation}\label{EQ:John}
  nU-p\xrightarrow{d} N(\nu_4-2,4).
\end{equation}
Since $\nu_4=3$ for normal distribution, it shows that the existing results confirm with each other. In this paper, we extend the above result one step further, i.e. consider the asymptotic behavior of the John's test statistic under the ultra-dimensional $p\gg n$ setting. We find that this test statistic possesses a remarkable {\it dimension-proof} property, which shows that under the $(n,p)$-asymptotic,
 the limit in \eqref{EQ:John} still holds when $ \lim_{n\rightarrow \infty} p/n=\infty$.
\noindent
This {\it dimension-proof} property of John's test makes it a very competitive candidate for sphericity testing regardless of $p,n$.

       Related methods have also been proposed in the literature for the high dimensional sphericity test. Noteworthy work include \citet{Schott05} where a test statistic based on the logarithm of the norm of sample correlation matrix under $(n,p)$-asymptotic has been well studied. Yet multivariate normality assumption has been assumed in this paper. Similarly in \citet{Fisher11}, a novel test statistic utilizing the ratio of the fourth and second arithmetic means of the sample covariance matrix is developed under the $p/n\rightarrow c$, $(n,p)$-asymptotic with normality restriction. \citet{Srivastava05} considered the ratio of arithmetic means of the eigenvalues of sample covariance matrix in the normal case when $n=O(p^{\delta}), \delta>0$, $n,p\rightarrow\infty$ and \citet{Srivastava11} further proved the robustness of this test statistic against non-normality assumption irrespective of either $n/p\rightarrow 0$ or $n/p\rightarrow\infty$. However, their results are only applicable under some specified factorized settings, which makes it less general than John's test.  \citet{S.X.Chen10} developed a high-dimensional test based on the John's test, however this test is very time-consuming (See Section \ref{simsec}). \citet{Zou13} considered the multivariate-sign-based covariance matrices to construct robust test for sphericity and significantly enhanced test performance when the non-normality is severe, particularly for heavy tailed distributions. In their paper the asymptotic distributions of the test statistic when $p=O(n^2)$ is derived. \citet{Srivastava06} studied a quasi-likelihood ratio test under the $n=O(p^{\delta}), ~0<\delta<1$, $n,p\rightarrow\infty$ asymptotic in the normal case, while in this paper, the normality assumption is released and results are discussed under a wider range of $(n,p)$-asymptotic.    These tests are compared in the simulation studies of the paper in Section \ref{simsec}. 

The rest of the paper is organized as follows. Section \ref{mainsec} discusses the asymptotic behavior of the John's test statistic and the quasi-LRT test statistic under the ultra-dimensional setting. Empirical sizes and powers of these two tests and other methods are compared under various scenarios. Section \ref{Powersec} \re{presented theoretical results for power of John's test and quasi-LRT test and testified these results with simulations}. Section \ref{empsec} concludes. Some technique lemmas and related proofs are displayed in the Appendix \ref{lemsec}.

\section{New tests and their asymptotic distributions}\label{mainsec}
\subsection{Preliminary Knowledge}
For any $n\times n$ Hermitian matrix $M$ with real eigenvalues $\lambda_1,\cdots,\lambda_n$, the empirical spectral distribution (ESD for short) of $M$ is defined by $F^M=n^{-1}\sum_{j=1}^n \delta_{\lambda_j}$,
where $\delta_{a}$ denotes the Dirac mass at $a$. The Stieltjes transform of any distribution $G$ is defined as
\[m_G(z)=\int \cfrac{1}{x-z}dG(x),~\mathfrak{I}(z)>0,\]
where $\mathfrak{I}(z)$ stands for the imaginary part of $z$.

Consider the re-normalized sample covariance matrix $A=\sqrt{\cfrac{p}{n}}\lb\cfrac{1}{p}X'X-I_n\rb$, where $X=(x_{ij})_{p\times n}$ and $x_{ij},i=1,\cdots,p,~j=1,\cdots,n$ are i.i.d. real random variables with mean zero and variance one, $I_n$ is the identity matrix of order $n$. It's known that under the ultra-dimensional setting \citep{Bai88}, with probability one, the ESD of matrix $A$, $F^A$ converges to the semicircle law $F$ with density
\[
F'(x)=\left\{
\begin{array}{ll}
  \cfrac{1}{2\pi}\sqrt{4-x^2},&~\mbox{if } |x|\leq 2,\\
  0,&~\mbox{if } |x|>2.
\end{array}
\right.
\]
\noindent
We denote the Stieltjes transform of the semicircle law $F$ by $m(z)$. Let $\mathscr{S}$ denote any open region on the complex plane including $[-2,2]$, the support of $F$ and $\mathscr{M}$ be the set of functions which are analytic on $\mathscr{S}$.
For any $f\in \mathscr{M}$, denote
\begin{equation}\label{MeanTerm}
 G_n(f)\triangleq n\int_{-\infty}^{+\infty}f(x) d \lb F^A(x)-F(x)\rb-\cfrac{n}{2\pi i}\oint_{|m|=\rho}f\lb -m-m^{-1}\rb \chi_n(m)\cfrac{1-m^2}{m^2} \operatorname{d}\!m,
\end{equation}
\noindent
where
\[\chi_n(m)\triangleq \cfrac{-\mathcal{B}+\sqrt{\mathcal{B}^2-4\mathcal{A}\mathcal{C}}}{2\mathcal{A}},~\mathcal{A}=m-\sqrt{\cfrac{n}{p}}(1+m^2),\]
\[\mathcal{B}=m^2-1-\cfrac{n}{p}m(1+2m^2),~\mathcal{C}=\cfrac{m^3}{n}\lb\cfrac{m^2}{1-m^2}+\nu_4-2\rb-\sqrt{\cfrac{n}{p}}m^4,\]
$\nu_4=\mathbb{E}X_{11}^4$ and $\sqrt{\mathcal{B}^2-4\mathcal{AC}}$ is a complex number whose imaginary part has same sign as that of $\mathcal{B}$. The integral's contour is taken as $|m|=\rho$ with $\rho<1$. \citet{ChenPan13} gives a calibration in advance for the mean correction term in \eqref{MeanTerm}, where only $\mathcal{C}$ is replaced with
\[\mathcal{C}^{\rm{Calib}}=\cfrac{m^3}{n}\left[\nu_4-2+\cfrac{m^2}{1-m^2}-2(\nu_4-1)m\sqrt{\cfrac{n}{p}}\right]-\sqrt{\cfrac{n}{p}}m^4\]
while others remain the same.

The central limit theorem (CLT) of linear functions of eigenvalues of the re-normalized sample covariance matrix $A$ when the dimension $p$ is much larger than the sample size $n$ derived by \citet{ChenPan13} is stated as follows.

\begin{theorem}\label{CLT}
  Suppose that
  \begin{itemize}
    \item [(a)] ${\bf X}=(x_{ij})_{p\times n}$ where $\{x_{ij}:~i=1,\cdots,p;~j=1,\cdots,n\}$ are i.i.d. real random variables with $\mathbb{E}X_{11}=0$, $\mathbb{E}X_{11}^2=1$ and $\nu_4=\mathbb{E}X_{11}^4<\infty$.
    \item [(b)] $n/p\rightarrow 0$  as $n\rightarrow\infty$.
  \end{itemize}
  Then, for any $f_1,\cdots,f_k\in \mathscr{M}$, the finite dimensional random vector $\lb G_n(f_1),\cdots,G_n(f_k)\rb$ converges weakly to a Gaussian vector $\lb Y(f_1),\cdots, Y(f_k)\rb$ with mean function $\mathbb{E}Y(f)=0$ and covariance function
  \begin{align}\label{cov}
    cov\lb Y(f_1), Y(f_2)\rb&=(\nu_4-3)\Phi_1(f_1)\Phi_1(f_2)+2\sum_{k=1}^{\infty}k\Phi_k(f_1)\Phi_k(f_2)\\ \nonumber
    &=\frac{1}{4\pi^2}\int_{-2}^2\int_{-2}^2f_1'(x)f_2'(y)H(x,y)\operatorname{d}\! x \d y
  \end{align}
  where
  \[\Phi_k(f)\triangleq\cfrac{1}{2\pi}\int_{-\pi}^{\pi}f(2\cos\theta)e^{ik\theta}\d \theta=\cfrac{1}{2\pi}\int_{-\pi}^{\pi}f(2\cos\theta)\cos k\theta \d\theta,\]
  \[H(x,y)=(\nu_4-3)\sqrt{4-x^2}\sqrt{4-y^2}+2\log\lb \cfrac{4-xy+\sqrt{(4-x^2)(4-y^2)}}{4-xy-\sqrt{(4-x^2)(4-y^2)}}\rb.\]
\end{theorem}

\noindent
The proofs of the main theorems in this paper are based on two lemmas derived from this CLT. Notice that the limiting covariance functions in \eqref{cov} has been first established in \citet{BaiYao05} for Wigner matrices.
\begin{lemma}\label{lemma}
  Let $\{\lambda_i,~1\leq i\leq n\}$ be eigenvalues of the matrix $A=\sqrt{\cfrac{p}{n}}\lb\cfrac{1}{p}X'X-I_n\rb$, where $X$ satisfies the assumptions in Theorem \ref{MainThm}, then as $p/n\rightarrow\infty$, $n\rightarrow\infty$,
  \[
  \left(
  \begin{array}{c}
    \sum_{i=1}^n\lambda_i^2-n-(\nu_4-2)\\
    \sum_{i=1}^n\lambda_i
  \end{array}
  \right)\xrightarrow{d} N\left(\left(
  \begin{array}{c}
    0\\
    0
  \end{array}
  \right),~
  \left(
    \begin{array}{cc}
      4 & 0 \\
      0 & \nu_4-1 \\
    \end{array}
  \right)
  \right).
  \]
  \end{lemma}

  \begin{lemma}\label{LRTlem}
  Let $\{\lambda_i,~1\leq i\leq n\}$ be eigenvalues of matrix $A=\sqrt{\cfrac{p}{n}}\lb\cfrac{1}{p}X'X-I_n\rb$, where $X$ satisfies the assumptions in Theorem \ref{MainThm2}, then as $p/n\rightarrow\infty$, $n\rightarrow\infty$,
\[
  \left(
  \begin{array}{c}
    \sum_{i=1}^n\lambda_i\\
    \sqrt{\frac pn}\sum_{i=1}^n\log \lb1+\lambda_i\sqrt{\frac np}\rb +\frac{1}{2}\sqrt{\frac{n^3}{p}}+\frac{n^2}{6p}\sqrt{\frac{n}{p}}+\frac{\nu_4-2}{2}\sqrt{\frac np}
  \end{array}
  \right)=\xi_n+o_p(1), \]
  where
  \[\xi_n\sim N\left(\left(
  \begin{array}{c}
    0\\
    0
  \end{array}
  \right),~
  \left(
    \begin{array}{cc}
      \nu_4-1 & (\nu_4-1)\lb 1+\frac np\rb \\
       (\nu_4-1)\lb1+\frac np\rb & \nu_4-1+\frac np(2\nu_4-1) \\
    \end{array}
  \right)
  \right).
  \]
\end{lemma}

The proofs of these two lemma are postponed to Appendix \ref{lemsec}.

\subsection{John's Test}
 Consider John's test statistic $U$ defined in \eqref{John} based on
 eigenvalues of the $p-$dimensional sample covariance matrix $S=\frac{1}{n}XX'$. Here we assume that the $X_j's$ in $X$ have representation $X_j=\Sigma^{1/2}Z_j$, where $\{Z_1,\cdots,Z_n\}=\{z_{ij}\}_{1\leq i\leq p,1\leq j\leq n}$ is a $p\times n$ matrix with i.i.d. entries $z_{ij}$ satisfying $\mathbb{E}(z_{ij})=0$, $\mathbb{E}(z_{ij}^2)=1$. It can be seen that, under the null hypothesis $H_0$, the John's test statistic is independent from the scale parameter $\sigma^2$. Therefore, we assume w.l.o.g. $\sigma^2=1$ when we derive the null distribution of the test statistic. In other words, under $H_0$, we assume in the rest of this paper that sample vectors $\{x_{ij}\}_{1\leq i\leq p, 1\leq j\leq n}$ satisfy $\mathbb{E}(x_{ij})=0$, $\mathbb{E}(x_{ij}^2)=1$, $\mathbb{E}(|x_{ij}|^4)=\nu_4<+\infty$.  The first main result of this paper is the following.

\begin{theorem}\label{MainThm}
  Assume $X=\{x_{ij}\}_{p\times n}$ are i.i.d. satisfying $\mathbb{E}(x_{ij})=0$, $\mathbb{E}(x_{ij}^2)=1$, $\mathbb{E}|x_{ij}|^4=\nu_4<\infty$, then when $p/n\rightarrow\infty$, $n\rightarrow\infty$,
  \[nU-p \xrightarrow{d} N(\nu_4-2,4).\]
\end{theorem}

Similarly with this theorem, \citet{Wang13} shows that if $\{x_{ij}\}$ are i.i.d. with $\mathbb{E}x_{ij}=0$, $\mathbb{E}|x_{ij}|^2=1$, $\nu_4\triangleq \mathbb{E}|x_{ij}|^4<\infty$, then when $n,p\rightarrow\infty$, $\lim_{n\rightarrow \infty} p/n\rightarrow c\in(0,+\infty)$,
\begin{equation*}
  nU-p\xrightarrow{d} N(\nu_4-2,4).
\end{equation*}
 It indicates that as long as $X=\{x_{ij}\}_{p\times n}$ are i.i.d with zero mean, unit variance and finite fourth order moment, John's test statistic $nU-p$ has a consistent limiting distribution $N(\nu_4-2,4)$, regardless of normality, under any $(n,p)$-asymptotic, $n/p\rightarrow[0,\infty)$. Therefore, the powerful {\it dimension-proof} property assigns John's test top priority when little information about the data is known before implementing sphericity test.

\vspace{0.4cm}
\noindent
The proof of  Theorem \ref{MainThm} is based on Lemma \ref{lemma}.

\noindent

\begin{proof}
  Denote the eigenvalues of $p\times p$ matrix $S_n=\frac{1}{n}XX'$ in descending order by $l_i(1\leq i\leq p)$, and the eigenvalues of $n\times n$ matrix $A=\sqrt{\cfrac{p}{n}}\lb\cfrac{1}{p}X'X-I_n\rb$ by $\lambda_i(1\leq i\leq n)$. Since $p>n$, $S_n$ has $p-n$ zero eigenvalues and the remaining $n$ non-zero eigenvalues $l_i(1\leq i\leq n)$ are related with $\lambda_i(1\leq i\leq n)$ eigenvalues of $A$ as
  \[\sqrt{\cfrac{p}{n}}\lambda_i+\cfrac{p}{n}=l_i,~~1\leq i\leq n.\]
 We have, for John's test statistic
   \begin{align*}
    U&=\left.\lb\cfrac{1}{p}\sum_{i=1}^n\cfrac{p^2}{n^2}\lb\sqrt{\cfrac{p}{n}}\lambda_i+1\rb^2\rb \middle/ \lb\cfrac{1}{p}\sum_{i=1}^n\cfrac{p}{n}\lb\sqrt{\cfrac{n}{p}}\lambda_i+1\rb\rb^2\right.-1\\
    &=\cfrac{\sum_{i=1}^n\lambda_i^2+2\sqrt{\cfrac{p}{n}}\sum_{i=1}^n\lambda_i+p}{\lb\sqrt{\cfrac{1}{p}}\sum_{i=1}^n\lambda_i+\sqrt{n}\rb^2}-1,
  \end{align*}
  Define the function $G(u,v)=\cfrac{u+2v\sqrt{\frac{p}{n}}+p}{(\sqrt{\frac{1}{p}}v+\sqrt{n})^2}-1$,
  then John's test statistic can be written as
  \[U=G\lb u=\sum_{i=1}^n\lambda_i^2,v=\sum_{i=1}^n\lambda_i\rb.\]
  According to Lemma \ref{lemma}, when $p/n\rightarrow\infty$, $n\rightarrow\infty$,
\[
\left(
\begin{array}{c}
\sum_{i=1}^n\lambda_i^2-n-(\nu_4-2)\\
\sum_{i=1}^n\lambda_i
\end{array}
\right)\xrightarrow{d} N\left(
\lb
\begin{array}{c}
  0\\
  0
\end{array}
\rb,~\lb
\begin{array}{cc}
 4&0\\
 0&\nu_4-1
\end{array}
\rb
\right).
\]
Then by the Delta Method,
\[
n\lb U-\left.G(u,v)\right\rvert_{u=n+\nu_4-2,v=0}\rb = \xi_n+o_p(1)\]
where
\[ \xi_n\sim N\lb 0,~
n^2 \nabla G\lb
\begin{array}{cc}
 4&0\\
 0&\nu_4-1
\end{array}
\rb
\nabla G'
\rb,
\]
and $\nabla G=\left.\lb\cfrac{\partial U}{\partial u},\cfrac{\partial U}{\partial \nu}\rb\right\rvert_{u=n+\nu_4-2,v=0}$ is the corresponding gradient vector.

\noindent

We have, for $(u,v)=(n+\nu_4-2,0)$,
\[
G=\cfrac{p}{n}+\cfrac{\nu_4-2}{n},\]
and
\[
\nabla G
\lb
\begin{array}{cc}
 4&0\\
 0&\nu_4-1
\end{array}
\rb
\nabla G'=\cfrac{4}{n^2}+\cfrac{4(\nu_4-1)}{np}\lb1+\cfrac{\nu_4-2}{n}\rb^2.
\]
The conclusion thus follows.
\end{proof}

\subsection{Quasi-likelihood ratio test}

Consider the Quasi-LRT statistic $\mathcal{L}_n$ in \eqref{quasiLRT}
based on the eigenvalues of $n-$dimensional matrix $\frac{1}{p}X'X$, which are also proportional to the non-null eigenvalues of $p-$dimensional sample covariance matrix $\frac{1}{n}XX'$. Similarly with John's test statistic, it can be seen that, under the null hypothesis $H_0$, the $\mathcal{L}_n$ statistic is independent of the scale parameter $\sigma^2$. Therefore, we again  assume w.l.o.g. $\sigma^2=1$ when we derive the null distribution of the test statistic.  The second main result of this paper is the following theorem.

\begin{theorem}\label{MainThm2}
  Assume $X=\{x_{ij}\}_{p\times n}$ are i.i.d. satisfying $\mathbb{E}(x_{ij})=0$, $\mathbb{E}(x_{ij}^2)=1$, $\mathbb{E}|x_{ij}|^4=\nu_4<\infty$, then when $p/n\rightarrow\infty$, $n\rightarrow\infty$,
  \begin{equation}\label{Maineq2}
    \mathcal{L}_n-\frac{n}{2}-\frac{n^2}{6p}-\frac{\nu_4-2}{2} \xrightarrow{d} N\lb 0,1\rb.
  \end{equation}
\end{theorem}

Recall the classic LRT when $H_0$ holds and $p$ is fixed while $n\rightarrow \infty$, if the population is Gaussian, the test statistic
\begin{equation*}
  -2\log L_n=n\log\lb\frac{\overline{l}^p}{\prod_{i=1}^p l_i}\rb\xrightarrow{d} \chi^2_{\frac{1}{2}p(p+1)-1},
\end{equation*}
where $\{l_i\}_{1\leq i\leq p}$ are the eigenvalues of $p-$dimensional sample covariance matrix $\frac1nXX'$. Here we notice that $n/p\rightarrow \infty$.

By interchanging the role of $n$ and $p$, which is feasible under $H_0$, it can be seen that when $n$ fixed and $p/n\rightarrow \infty$, the test statistic
\[-2\log L_p=p\log\lb\frac{\overline{l}^n}{\prod_{i=1}^n l_i}\rb\xrightarrow{d} \chi^2_{\frac{1}{2}n(n+1)-1},\]
$\{l_i\}_{1\leq i\leq n}$ are the eigenvalues of $n-$dimensional sample covariance matrix $\frac1pX'X$. Note that $\lb-2\log L_p\rb/n$ coincides with our Quasi-LRT statistic $\mathcal{L}_n$. Heuristically, if next we let $n\rightarrow\infty$, then
\[\cfrac{\chi^2_{\frac{1}{2}n(n+1)-1}}{n}-\frac{n+1}{2}\xrightarrow{d} N\lb 0,1\rb,\]
which is nothing but \eqref{Maineq2} applied to the normal case $(\nu_4=3)$ with fixed $n$ and $p\rightarrow\infty$. Therefore, the classical LRT can be thought \re{of} as a particular ``finite-dimensional" instance of the general limit of \eqref{Maineq2} for the Quasi-LRT, that is, Theorem \ref{MainThm2} covers a wide range of ``large p, small n" situations.

\vspace{0.4cm}
\noindent
The proof of Theorem \ref{MainThm2} is based on lemma \ref{LRTlem}.


\begin{proof}
  Denote the eigenvalues of $n\times n$ matrix $\frac{1}{p}X'X$ in descending order by $\tilde{l}_i(1\leq i\leq n)$, and eigenvalues of $n\times n$ matrix $A=\sqrt{\frac{p}{n}}\lb\frac{1}{p}X'X-I_n\rb$ by $\lambda_i(1\leq i\leq n)$. These eigenvalues are related as
  \[\sqrt{\frac{n}{p}}\lambda_i+1=\tilde{l}_i,~~1\leq i\leq n.\]
We have, for the Quasi-LRT test statistic
 \begin{align*}
    ~&~\mathcal{L}_n-\frac n2-\frac{n^2}{6p}-\frac{\nu_4-2}{2}\\
    =&~\frac{p}{n}\log\left[\left.\lb\frac{1}{n}\sum_{i=1}^n\tilde{l}_i\rb^n\middle/\prod_{i=1}^n\tilde{l}_i\right.\right]-\frac n2-\frac{n^2}{6p}-\frac{\nu_4-2}{2}\\ =&~p\log \lb1+\sqrt{\frac np}\lb\frac 1n\sum_{i=1}^n\lambda_i\rb\rb-\frac pn\sum_{i=1}^n\log\lb 1+\sqrt{\frac{n}{p}}\lambda_i\rb-\frac n2-\frac{n^2}{6p}-\frac{\nu_4-2}{2}
    \end{align*}
  Define the function
  $$G(u,v)=p\log \lb1+\sqrt{\frac np}\lb\frac 1n u\rb\rb-\sqrt{\frac pn}v,$$
  then the Quasi-LRT test statistic can be written as
  \[\mathcal{L}_n-\frac n2-\frac{n^2}{6p}-\frac{\nu_4-2}{2}=G\lb u=\sum_{i=1}^n\lambda_i,v=\sqrt{\frac pn}\sum_{i=1}^n\log \lb1+\lambda_i\sqrt{\frac np}\rb +\frac{1}{2}\sqrt{\frac{n^3}{p}}+\frac{n^2}{6p}\sqrt{\frac{n}{p}}+\frac{\nu_4-2}{2}\sqrt{\frac np}\rb.\]
  According to Lemma \ref{LRTlem}, when $p/n\rightarrow\infty$, $n\rightarrow\infty$,
  \[
  \left(
  \begin{array}{c}
    \sum_{i=1}^n\lambda_i\\
    \sqrt{\frac pn}\sum_{i=1}^n\log \lb1+\lambda_i\sqrt{\frac np}\rb +\frac{1}{2}\sqrt{\frac{n^3}{p}}+\frac{n^2}{6p}\sqrt{\frac{n}{p}}+\frac{\nu_4-2}{2}\sqrt{\frac np}
  \end{array}
  \right)=\xi_n+o_p(1), \]
  where
  \[\xi_n\sim N\left(\left(
  \begin{array}{c}
    0\\
    0
  \end{array}
  \right),~
  \left(
    \begin{array}{cc}
      \nu_4-1 & (\nu_4-1)\lb 1+\frac np\rb \\
       (\nu_4-1)\lb1+\frac np\rb & \nu_4-1+\frac np(2\nu_4-1) \\
    \end{array}
  \right)
  \right).
  \]

Then by the Delta Method,
\[
\mathcal{L}_n-\frac{n}{2}-\frac{n^2}{6p}-\frac{\nu_4-2}{2}-\left.G(u,v)\right\rvert_{u=0,v=0}\xrightarrow{d} N\lb 0,~
\nabla G\left(
    \begin{array}{cc}
      \nu_4-1 & (\nu_4-1)\lb 1+\frac np\rb \\
       (\nu_4-1)\lb1+\frac np\rb & \nu_4-1+\frac np(2\nu_4-1) \\
    \end{array}
  \right)
\nabla G'
\rb,
\]
where $\nabla G=\left.\lb\cfrac{\partial U}{\partial u},\cfrac{\partial U}{\partial \nu}\rb\right\rvert_{u=0,v=0}$ is the corresponding gradient vector.

\noindent
We have, for $(u,v)=(0,0)$, $G=0$ and
\[
\nabla G\left(
    \begin{array}{cc}
      \nu_4-1 & (\nu_4-1)\lb 1+\frac np\rb \\
       (\nu_4-1)\lb1+\frac np\rb & \nu_4-1+\frac np(2\nu_4-1) \\
    \end{array}
  \right)\nabla G'
=1.
\]
Therefore, when $p/n\rightarrow\infty$, $n\rightarrow\infty$,
 \[\mathcal{L}_n-\frac{n}{2}-\frac{n^2}{6p}-\frac{\nu_4-2}{2} \xrightarrow{d} N\lb 0,1\rb.\]
\end{proof}

\subsection{Simulation Studies}\label{simsec}

In order to further explore the finite sample behavior of John's sphericity test when dimension $p$ is significantly larger than the sample size $n$, Monte Carlo simulations are implemented in this session to evaluate the size and power of John's Sphericity Test. Test statistic proposed by \citet{S.X.Chen10} is also considered for comparison.

In the simulation, without loss of generality, we conduct the sphericity test with $\sigma^2=1$. To find the empirical sizes of these two tests, we consider two different scenarios to generate sample data:
\begin{itemize}
  \item [(1)]  $\{X_{j}\},~1\leq j\leq n$ i.i.d $p$-dimensional random vector generated from multivariate normal population $N(0,I_p)$, $\mathbb{E}x^4_{ij}=\nu_4=3$;
  \item [(2)] $\{x_{ij},~1\leq i\leq p,~1\leq j\leq n\}$ i.i.d follow $Gamma(4,2)-2$ distribution, then $\mathbb{E}x_{ij}=0$, $\mathbb{E}x^2_{ij}=1$, $\mathbb{E}x^4_{ij}=\nu_4=4.5$.
\end{itemize}
We set sample size $n=64$, dimension $p=320,640,960,1280,1600,2400,3200$ in order to understand the effect of an increasing dimension. The nominal test level is $\alpha=0.05$. For each pair of $(p,n)$, 10000 replications are used to get the empirical size.

For John's test, we reject $H_0$ if $nU-p$ exceeds the $5\%$ upper quantile of $N(\nu_4-2,4)$ distribution. For Quasi-LRT test, we reject $H_0$ if $\mathcal{L}_n-\frac n2-\frac{n^2}{6p}-\frac{\nu_4-2}{2}$ exceeds the $5\%$ upper quantile of $N(0,1)$ distribution.

As for the test in \citet{S.X.Chen10}, the test statistic is defined as follows:
\[U_n=p\lb\cfrac{T_{2,n}}{T_{1,n}^2}\rb-1,\]
where
\[ T_{1,n}=\cfrac{1}{n}\sum_{i=1}^nX_i'X_i-\cfrac{1}{P^2_n}\sum_{i\neq j}X_i'X_j,\]
\[T_{2,n}=\cfrac{1}{P^2_n}\sum_{i\neq j}(X_i'X_j)^2-\cfrac{2}{P^3_n}\sum_{i,j,k}^*X_i'X_jX_j'X_k+\cfrac{1}{P^4_n}\sum_{i,j,k,l}^*X_i'X_jX_k'X_l,\]
where $P_n^r=n!/(n-r)!$, $\sum^*$ denotes summation over mutually different indices. Then we reject $H_0$ if $nU_n$ exceeds the $5\%$ upper quantile of $N(0,4)$ distribution.

For the test in \citet{Srivastava11}(Sri for short), the test statistic is defined as follows:
\[W_n=\frac{n}{2}\cdot \left[\cfrac{c_n\cdot\frac 1p\left[tr S^2-\frac 1n(tr S)^2\right]}{\lb\frac 1p trS\rb^2}-1\right]\]
where $S=\frac 1nXX'$, $c_n=\frac{n^2}{(n-1)(n+2)}$. According to the limiting distribution of $W_n$, we reject $H_0$ if $W_n$ exceeds the $5\%$ upper quantile of $N(0,1)$ distribution.
As for empirical powers, we generate sample data from two alternatives:
\begin{itemize}
  \item [-] {\bf Power 1: } $\Sigma$ is diagonal with half of its diagonal elements 0.5 and half 1. This power scenario is denoted by Power 1;
  \item [-] {\bf Power 2: }  $\Sigma$ is diagonal with $1/4$ of its diagonal elements 0.5 and $3/4$ equal to 1. This power scenario is denoted by Power 2.
\end{itemize}

Table 1 reports the empirical sizes and powers of two tests for Gaussian data. Table 2 is for Non-Gaussian data.

\begin{table}[!h]\label{table1}
\caption{Scenario 1 for Gaussian Data}
\centering
\begin{tabular}{|c|cccc|cccc|cccc|}
\hline
 & \multicolumn{4}{c|}{Size} & \multicolumn{4}{c|}{Power1} & \multicolumn{4}{c|}{Power2}\tabularnewline
$\left(p,n\right)$ & Sri & Chen & John & QLRT & Sri & Chen & John & QLRT & Sri & Chen & John & QLRT\tabularnewline
\hline
$\left(320,64\right)$ & 0.048 & 0.0539 & 0.0492 & 0.0998 & 0.9571 & 0.9532 & 0.958 & 0.9777 & 0.6155 & 0.6117 & 0.6194 & 0.7352\tabularnewline
$\left(640,64\right)$ & 0.0504 & 0.0538 & 0.0515 & 0.0668 & 0.9595 & 0.9542 & 0.9602 & 0.9638 & 0.6089 & 0.6065 & 0.6128 & 0.6562\tabularnewline
$\left(960,64\right)$ & 0.0532 & 0.0581 & 0.0544 & 0.062 & 0.9598 & 0.9569 & 0.9604 & 0.9647 & 0.6201 & 0.6144 & 0.6231 & 0.6482\tabularnewline
$\left(1280,64\right)$ & 0.0519 & 0.0603 & 0.053 & 0.0568 & 0.9609 & 0.9569 & 0.9615 & 0.9656 & 0.6076 & 0.6043 & 0.6129 & 0.6256\tabularnewline
$\left(1600,64\right)$ & 0.0529 & 0.0571 & 0.0539 & 0.0593 & 0.9583 & 0.9539 & 0.9588 & 0.9627 & 0.6194 & 0.6146 & 0.6231 & 0.6378\tabularnewline
$\left(2400,64\right)$ & 0.0493 & 0.0536 & 0.0501 & 0.0506 & 0.9588 & 0.9542 & 0.9591 & 0.9615 & 0.6171 & 0.6099 & 0.621 & 0.6291\tabularnewline
$\left(3200,64\right)$ & 0.0472 & 0.0538 & 0.0481 & 0.0503 & 0.9617 & 0.9576 & 0.9624 & 0.9625 & 0.6212 & 0.619 & 0.6251 & 0.6301\tabularnewline
\hline
\end{tabular}
%
%
%
%
%
%
%
\end{table}

\begin{table}[!h]\label{table2}
\caption{Scenario 2 for Non-Gaussian Data}
\centering
\begin{tabular}{|c|cccc|cccc|cccc|}
\hline
 & \multicolumn{4}{c|}{Size} & \multicolumn{4}{c|}{Power1} & \multicolumn{4}{c|}{Power2}\tabularnewline
$\left(p,n\right)$ & Sri & Chen & John & QLRT & Sri & Chen & John & QLRT & Sri & Chen & John & QLRT\tabularnewline
\hline
$\left(320,64\right)$ & 0.1828 & 0.0584 & 0.0566 & 0.1084 & 0.9909 & 0.9476 & 0.9538 & 0.9701 & 0.8374 & 0.6044 & 0.6196 & 0.7299\tabularnewline
$\left(640,64\right)$ & 0.1875 & 0.0594 & 0.0598 & 0.0735 & 0.9927 & 0.9566 & 0.9603 & 0.9653 & 0.8379 & 0.6051 & 0.6201 & 0.6601\tabularnewline
$\left(960,64\right)$ & 0.1869 & 0.058 & 0.0551 & 0.0631 & 0.9923 & 0.9524 & 0.9589 & 0.9608 & 0.8394 & 0.6121 & 0.6298 & 0.6502\tabularnewline
$\left(1280,64\right)$ & 0.1856 & 0.057 & 0.0517 & 0.0605 & 0.9927 & 0.9529 & 0.9599 & 0.962 & 0.8483 & 0.6133 & 0.6206 & 0.6416\tabularnewline
$\left(1600,64\right)$ & 0.1811 & 0.0555 & 0.0536 & 0.058 & 0.9925 & 0.9557 & 0.9622 & 0.9642 & 0.8433 & 0.6143 & 0.633 & 0.6407\tabularnewline
$\left(2400,64\right)$ & 0.179 & 0.0581 & 0.0533 & 0.0564 & 0.991 & 0.9497 & 0.9567 & 0.9577 & 0.8425 & 0.611 & 0.6261 & 0.6304\tabularnewline
$\left(3200,64\right)$ & 0.1757 & 0.0518 & 0.0503 & 0.0522 & 0.9909 & 0.9529 & 0.961 & 0.9611 & 0.8413 & 0.6143 & 0.6266 & 0.6319\tabularnewline
\hline
\end{tabular}
\end{table}

It can be seen from the above results that both John's test and QLRT perform well with respect to sizes and powers. Empirical powers under Power 1 are in general higher than under Power 2 because of more significant difference between $H_0$ and $H_1$. John's test performs slightly better than Chen's method. In all tested scenarios, the QLRT dominates the other two tests in term of power even though the difference is quite marginal. Srivastava's test performs slightly below John's test in the Gaussian case and still suffers from non-normality with non-negligible bias. Furthermore, we have recorded the execution time of these two tests within different scenarios and we find that Chen's method is more time-consuming due to more complicated computations.

 \section{Power of the tests}\label{Powersec}

In this section we study the asymptotic power of the two tests. To begin with, some preliminary knowledge is introduced as follows.

  \subsection{Preliminary knowledge}

Consider the re-normalized sample covariance matrix
$$\widetilde{A}=\sqrt{\cfrac{1}{n}}\lb \cfrac{1}{\sqrt{\tr (\Sigma_p^2)}} Z'\Sigma_p Z-\cfrac{\tr (\Sigma_p)}{\sqrt{\tr (\Sigma_p^2)}}  I_n\rb,$$
 where $Z=(z_{ij})_{p\times n}$ and $z_{ij},i=1,\cdots,p,~j=1,\cdots,n$ are i.i.d. real random variables with mean zero and variance one, $I_n$ is the identity matrix of order $n$, $\Sigma_p$ is a sequence of $p\times p$ non-negative definite matrices with bounded spectral norm. Assume the following limit exist,
 \begin{itemize}
   \item[(a)] $\gamma=\lim_{p\rightarrow \infty} \frac{1}{p}\tr (\Sigma_p)$,
    \item[(b)] $\theta=\lim_{p\rightarrow \infty} \frac{1}{p}\tr (\Sigma_p^2)$,
     \item[(c)] $\omega=\lim_{p\rightarrow \infty} \frac{1}{p}\sum_{i=1}^p (\Sigma_{ii})^2$,
 \end{itemize}
it has been proven that, under the ultra-dimensional setting \citep{Bai88}, with probability one, the ESD of matrix $\widetilde{A}$, $F^{\widetilde{A}}$ converges to the semicircle law $F$ with density
\[
F'(x)=\left\{
\begin{array}{ll}
  \cfrac{1}{2\pi}\sqrt{4-x^2},&~\mbox{if } |x|\leq 2,\\
  0,&~\mbox{if } |x|>2.
\end{array}
\right.
\]
\noindent
We denote the Stieltjes transform of the semicircle law $F$ by $m(z)$. Let $\mathscr{S}$ denote any open region on the complex plane including $[-2,2]$, the support of $F$ and $\mathscr{M}$ be the set of functions which are analytic on $\mathscr{S}$.
For any $f\in \mathscr{M}$, denote
\begin{equation*}
 G_n(f)\triangleq n\int_{-\infty}^{+\infty}f(x) d \lb F^{\widetilde{A}}(x)-F(x)\rb-\sqrt{\cfrac{n^3}{p}}\Phi_3(f)
\end{equation*}
\noindent
where, for any positive integer $k$,
\[\Phi_k(f)=\frac{1}{2\pi}\int_{-\pi}^{\pi}f(2\cos(\theta))\cos(k\theta)\d \theta.\]

\noindent
Limiting theory of the test statistics under the alternative $H_1$ is based on a new CLT for linear statistics of $\widetilde{A}$, provided in \citet{LiYao16}, as follows.

\begin{theorem}\label{PowerCLT}
  Suppose that
  \begin{itemize}
    \item [(1)] ${Z}=(z_{ij})_{p\times n}$ where $\{z_{ij}:~i=1,\cdots,p;~j=1,\cdots,n\}$ are i.i.d. real random variables with $\mathbb{E}z_{ij}=0$, $\mathbb{E}z_{ij}^2=1$ and $\nu_4=\mathbb{E}z_{ij}^4<\infty$;
        \item[(2)] $\lb\Sigma_p\rb$ is a sequence of $p\times p$ non-negative definite matrices with bounded spectral norm and the following limit exist,
 \begin{itemize}
   \item[(a)] $\gamma=\lim_{p\rightarrow \infty} \frac{1}{p}\tr (\Sigma_p)$,
    \item[(b)] $\theta=\lim_{p\rightarrow \infty} \frac{1}{p}\tr (\Sigma_p^2)$,
     \item[(c)] $\omega=\lim_{p\rightarrow \infty} \frac{1}{p}\sum_{i=1}^p (\Sigma_{ii})^2$;
 \end{itemize}
    \item [(3)] $p/n\rightarrow \infty$  as $n\rightarrow\infty$, $n^3/p=O(1)$.
  \end{itemize}
  Then, for any $f_1,\cdots,f_k\in \mathscr{M}$, the finite dimensional random vector $\lb G_n(f_1),\cdots,G_n(f_k)\rb$ converges weakly to a Gaussian vector $\lb Y(f_1),\cdots, Y(f_k)\rb$ with mean function
  $$\mathbb{E}Y(f)=\frac{1}{4}\lb f(2)+f(-2)\rb-\frac{1}{2}\Phi_0(f)+\cfrac{\omega}{\theta}(\nu_4-3)\Phi_2(f),$$
   and covariance function
  \begin{align}\label{cov}
    cov\lb Y(f_1), Y(f_2)\rb&=\frac{\omega}{\theta}(\nu_4-3)\Phi_1(f_1)\Phi_1(f_2)+2\sum_{k=1}^{\infty}k\Phi_k(f_1)\Phi_k(f_2)\\ \nonumber
    &=\frac{1}{4\pi^2}\int_{-2}^2\int_{-2}^2f_1'(x)f_2'(y)H(x,y)\operatorname{d}\! x \d y,
  \end{align}
  where
  \[\Phi_k(f)\triangleq\cfrac{1}{2\pi}\int_{-\pi}^{\pi}f(2\cos\theta)e^{ik\theta}\d \theta=\cfrac{1}{2\pi}\int_{-\pi}^{\pi}f(2\cos\theta)\cos k\theta \d\theta,\]
  \[H(x,y)=\frac{\omega}{\theta}(\nu_4-3)\sqrt{4-x^2}\sqrt{4-y^2}+2\log\lb \cfrac{4-xy+\sqrt{(4-x^2)(4-y^2)}}{4-xy-\sqrt{(4-x^2)(4-y^2)}}\rb.\]
\end{theorem}

\noindent
The proofs of Theorem \ref{PowerThm1} and \ref{PowerThm2} about the power of the two test statistics are based on two lemmas derived from this CLT.

\begin{lemma}\label{Powerlem1}
  Let $\{\widetilde{\lambda}_i,~1\leq i\leq n\}$ be eigenvalues of matrix $\widetilde{A}=\sqrt{\frac{1}{n}}\lb \frac{1}{\sqrt{\tr (\Sigma_p^2)}} Z'\Sigma_p Z-\frac{\tr (\Sigma_p)}{\sqrt{\tr (\Sigma_p^2)}}  I_n\rb$, where $Z$, $\Sigma_p$ satisfies the assumptions in Theorem \ref{PowerCLT}, then
  \[
  \left(
  \begin{array}{c}
    \sum_{i=1}^n\widetilde{\lambda}_i^2-n-\lb \frac{\omega}{\theta}(\nu_4-3)+1\rb\\
    \sum_{i=1}^n\widetilde{\lambda}_i
  \end{array}
  \right)\xrightarrow{d} N\left(\left(
  \begin{array}{c}
    0\\
    0
  \end{array}
  \right),~
  \left(
    \begin{array}{cc}
      4 & 0 \\
      0 & \frac{\omega}{\theta}(\nu_4-3)+2 \\
    \end{array}
  \right)
  \right)
  \]
  as $p/n\rightarrow\infty$, $n\rightarrow\infty$, $n^3/p=O(1)$,
  \end{lemma}

  \begin{lemma}\label{Powerlem2}
  Let $\{\widetilde{\lambda}_i,~1\leq i\leq n\}$ be eigenvalues of matrix $\widetilde{A}=\sqrt{\frac{1}{n}}\lb \frac{1}{\sqrt{\tr (\Sigma_p^2)}} Z'\Sigma_p Z-\frac{\tr (\Sigma_p)}{\sqrt{\tr (\Sigma_p^2)}}  I_n\rb$, where $Z$, $\Sigma_p$ satisfies the assumptions in Theorem \ref{PowerCLT}, then
   \[
  \left(
  \begin{array}{c}
    \sum_{i=1}^n\widetilde{\lambda}_i\\
    \sqrt{\frac pn}\sum_{i=1}^n\log \lb\gamma+\widetilde{\lambda}_i\sqrt{\frac {n\theta}{p}}\rb -\sqrt{pn}\log(\gamma)+\frac{\theta}{2\gamma^2}\sqrt{\frac{n^3}{p}}+\lb\lb \frac{\theta^2}{2\gamma^4}-\frac{\theta\sqrt{\theta}}{3\gamma^3}\rb\frac{n^2}{p}+\frac{\theta}{2\gamma^2}+\frac{\omega}{2\gamma^2}(\nu_4-3)\rb\sqrt{\frac np}
  \end{array}
  \right)\]\[=\xi_n+o_p(1),
  \]
  where
  \[\xi_n\sim N\left(\left(
  \begin{array}{c}
    0\\
    0
  \end{array}
  \right),~
  \left(
    \begin{array}{cc}
      \frac{\omega}{\theta}\lb\nu_4-3\rb+2 & \lb \frac{\omega}{\theta}\lb\nu_4-3\rb+2\rb\lb\frac{\sqrt{\theta}}{\gamma}+\frac{\theta\sqrt{\theta}}{\gamma^3}\frac np\rb \\
       \lb \frac{\omega}{\theta}\lb\nu_4-3\rb+2\rb\lb\frac{\sqrt{\theta}}{\gamma}+\frac{\theta\sqrt{\theta}}{\gamma^3}\frac np\rb  & \frac{\lb \frac{\omega}{\theta}\lb\nu_4-3\rb+2\rb\theta }{\gamma^2}+\frac{\lb \frac{2\omega}{\theta}\lb\nu_4-3\rb+5\rb\theta^2n}{\gamma^4p} \\
    \end{array}
  \right)
  \right)
  \]
  as $p/n\rightarrow\infty$, $n\rightarrow\infty$, $n^3/p=O(1)$.
\end{lemma}

The proofs of these two lemma are postponed to Appendix \ref{lemsec}.

  \subsection{John's test}
  Suppose that an i.i.d. $p-$dimensional sample vectors $X_1,\cdots,X_n$ follow the multivariate distribution with covariance matrix $\Sigma_p$. To explore the power of John's test under the alternative hypothesis  $H_1: \Sigma_p\neq \sigma^2 I_p$,
we assume that the $X_j's$ in $X$ have representation $X_j=\Sigma_p^{1/2}Z_j$, so as $S=\frac{1}{n}\Sigma_p^{1/2}ZZ'\Sigma_p^{1/2}$, where $Z=\{Z_1,\cdots,Z_n\}=\{z_{ij}\}_{1\leq i\leq p,1\leq j\leq n}$ is a $p\times n$ matrix with i.i.d. entries $z_{ij}$ satisfying $\mathbb{E}(z_{ij})=0$, $\mathbb{E}(z_{ij}^2)=1$ and $\mathbb{E}(|z_{ij}|^4)=\nu_4<+\infty$. Then John's test statistic is   \[U=\frac{p^{-1}\sum_{i=1}^p(l_i-\overline{l})^2}{\overline{l}^2},\]
 where $\{l_i,~1\leq i\leq p\}$ are eigenvalues of the $p-$dimensional sample covariance matrix $S=\frac{1}{n}\Sigma_p^{1/2}ZZ'\Sigma_p^{1/2}$. The main result of the power of John's test is as follows.

\begin{theorem}\label{PowerThm1}
  Assume $X_1,\cdots,X_n$ are i.i.d. $p-$dimensional sample vectors follow multivariate distribution with covariance matrix $\Sigma_p$, $X=\Sigma_p^{1/2}Z$ where $Z=\{z_{ij}\}$ is a $p\times n$ matrix with i.i.d. entries $z_{ij}$satisfying $\mathbb{E}(z_{ij})=0$, $\mathbb{E}(z_{ij}^2)=1$, $\mathbb{E}|z_{ij}|^4=\nu_4<\infty$, $\Sigma_p$ is a sequence of $p\times p$ non-negative definite matrices with bounded spectral norm and the following limit exist,
 \begin{itemize}
   \item[(a)] $\gamma=\lim_{p\rightarrow \infty} \frac{1}{p}\tr (\Sigma_p)$,
    \item[(b)] $\theta=\lim_{p\rightarrow \infty} \frac{1}{p}\tr (\Sigma_p^2)$,
     \item[(c)] $\omega=\lim_{p\rightarrow \infty} \frac{1}{p}\sum_{i=1}^p (\Sigma_{ii})^2$,
 \end{itemize} then when $p/n\rightarrow\infty$, $n\rightarrow\infty$, $n^3/p=O(1)$,
    \[nU-p-\lb\cfrac{\theta}{\gamma^2}-1\rb n \xrightarrow{d} N\lb\cfrac{\theta+\omega(\nu_4-3)}{\gamma^2},~\cfrac{4\theta^2}{\gamma^4}\rb.\]
\end{theorem}

Note that the theorem above reveals the limit distribution of John's test statistic under alternative hypothesis $H_1$. Nevertheless, if let $\Sigma_p=\sigma^2I_p$, then $\gamma=\sigma^2$, $\theta=\omega=\sigma^4$, Theorem \ref{PowerThm1} reduces to Theorem \ref{MainThm}, which states the null distribution of John's test statistic under $H_0$. With the two limit distributions of John's test statistic under $H_0$ and $H_1$, power of the test is derived as below.

\begin{proposition}\label{Prop:PowerJohn}
  With the same assumptions as in Theorem \ref{PowerThm1}, when $p/n\rightarrow\infty,~ n\rightarrow\infty,~n^3/p=O(1)$,
  the power of John's test
   $$\beta_{\text{John}}(H_1)=1-\Phi\lb \frac{\gamma^2}{\theta}Z_{\alpha}+\frac{\gamma^2(\nu_4-2)-\theta-\omega(\nu_4-3)}{2\theta}+\frac{(\gamma^2-\theta)n}{2\theta}\rb\rightarrow 1,$$
   where  $\alpha$ is the nominal test level, $Z_\alpha$, $\Phi(\cdot)$ are the alpha upper quantile and cdf of standard normal distribution respectively.
\end{proposition}

\noindent
For John's test statistic $U$, under $H_0$,
\[nU-p\xrightarrow{d}N(\nu_4-2,4),\]
under $H_1$,
\[nU-p-\lb\frac{\theta}{\gamma^2}-1\rb n\xrightarrow{d}N\lb\frac{\theta+\omega(\nu_4-3)}{\gamma^2},\frac{4\theta^2}{\gamma^4}\rb,\]

\begin{align*}
  \beta_{\text{John}}(H_1)&=P\lb\left. \cfrac{nU-p-(\nu_4-2)}{2}>Z_\alpha\right| H_1 \rb\\
  &=P\lb \cfrac{nU-p-n\lb\frac{\theta}{\gamma^2}-1\rb-\frac{\theta+\omega(\nu_4-3)}{\gamma^2}}{\frac{2\theta}{\gamma^2}}>\cfrac{2Z_\alpha+(\nu_4-2)-n\lb\frac{\theta}{\gamma^2}-1\rb-\frac{ \theta+\omega(\nu_4-3)}{\gamma^2}}{\frac{2\theta}{\gamma^2}}  \rb\\
  &=1-\Phi\lb \frac{\gamma^2}{\theta}Z_{\alpha}+\frac{\gamma^2(\nu_4-2)-\theta-\omega(\nu_4-3)}{2\theta}+\frac{(\gamma^2-\theta)n}{2\theta}\rb,
  \end{align*}
  According to Jensen's inequality, $\gamma^2\leq \theta$ and equality holds only when $\Sigma_p=\sigma^2 I_p$, Proposition \ref{Prop:PowerJohn} thus follows.
\vspace{0.4cm}

\noindent
The proof of  Theorem \ref{PowerThm1} is based on Lemma \ref{Powerlem1}.

\noindent

\begin{proof}
  Denote the eigenvalues of $p\times p$ matrix $S_n=\frac{1}{n}XX'=\frac{1}{n}Z\Sigma_pZ'$ in descending order by $\{l_i, ~1\leq i\leq p\}$, and eigenvalues of $n\times n$ matrix $\widetilde{A}=\sqrt{\cfrac{1}{n}}\lb \cfrac{1}{\sqrt{\tr (\Sigma_p^2)}} Z'\Sigma_p Z-\cfrac{\tr (\Sigma_p)}{\sqrt{\tr (\Sigma_p^2)}}  I_n\rb$ by $\{\widetilde{\lambda}_i,~1\leq i\leq n\}$. Since $p>n$, $S_n$ has $p-n$ zero eigenvalues and the remaining $n$ non-zero eigenvalues $l_i$ are  related with $\widetilde{\lambda}_i$ as
  \[\sqrt{\frac{1}{n}\tr(\Sigma_p^2)}\widetilde{\lambda}_i+\frac{1}{n}\tr(\Sigma_p)=l_i,~~1\leq i\leq n.\]
  We have, for John's test statistic
 \begin{align*}
    U&=\left.\lb\cfrac{1}{p}\sum_{i=1}^n\lb\sqrt{\frac{1}{n}\tr(\Sigma_p^2)}\widetilde{\lambda}_i+\frac{1}{n}\tr(\Sigma_p)\rb^2\rb \middle/ \lb\cfrac{1}{p}\sum_{i=1}^n\lb\sqrt{\frac{1}{n}\tr(\Sigma_p^2)}\widetilde{\lambda}_i+\frac{1}{n}\tr(\Sigma_p)\rb\rb^2\right.-1\\
   &=\cfrac{\theta\sum_{i=1}^n\widetilde{\lambda}_i^2+2\gamma\sqrt{\frac{p\theta}{n}}\sum_{i=1}^n\widetilde{\lambda}_i+p\gamma^2}{\lb \sqrt{\frac{\theta}{p}}\sum_{i=1}^n\widetilde{\lambda}_i+\sqrt{n}\gamma\rb^2}-1.
  \end{align*}
  Define function $G(u,v)=\cfrac{\theta u+2 \gamma \sqrt{\frac{p\theta}{n}}v+p\gamma^2}{(\sqrt{\frac{\theta}{p}}v+\sqrt{n}\gamma)^2}-1$,
  then John's test statistic can be written as
  \[U=G\lb u=\sum_{i=1}^n\widetilde{\lambda}_i^2,v=\sum_{i=1}^n\widetilde{\lambda}_i\rb.\]
  According to Lemma \ref{Powerlem1}, when $p/n\rightarrow\infty$, $n\rightarrow\infty$, $n^3/p=O(1)$,
\[
  \left(
  \begin{array}{c}
    \sum_{i=1}^n\widetilde{\lambda}_i^2-n-\lb \frac{\omega}{\theta}(\nu_4-3)+1\rb\\
    \sum_{i=1}^n\widetilde{\lambda}_i
  \end{array}
  \right)\xrightarrow{d} N\left(\left(
  \begin{array}{c}
    0\\
    0
  \end{array}
  \right),~
  \left(
    \begin{array}{cc}
      4 & 0 \\
      0 & \frac{\omega}{\theta}(\nu_4-3)+2 \\
    \end{array}
  \right)
  \right)
  \]
Then by the Delta Method,
\[
n\lb U-\left.G(u,v)\right\rvert_{u=n+ \frac{\omega}{\theta}(\nu_4-3)+1,v=0}\rb\xrightarrow{d} N\lb 0,~
n^2\nabla G\lb
\begin{array}{cc}
      4 & 0 \\
      0 & \frac{\omega}{\theta}(\nu_4-3)+2 \\
    \end{array}
\rb
\nabla G'
\rb,
\]
where $\nabla G=\left.\lb\cfrac{\partial U}{\partial u},\cfrac{\partial U}{\partial \nu}\rb\right\rvert_{u=n+ \frac{\omega}{\theta}(\nu_4-3)+1,v=0}$ is the corresponding gradient vector.

\noindent

We have, for $(u,v)=\lb n+ \frac{\omega}{\theta}(\nu_4-3)+1,0\rb$,
\[G=\cfrac{p}{n}+\cfrac{\theta}{\gamma^2}-1+\cfrac{\lb\omega(\nu_4-3)+\theta \rb}{n\gamma^2},\]
and
\[
\nabla G
\lb
\begin{array}{cc}
      4 & 0 \\
      0 & \frac{\omega}{\theta}(\nu_4-3)+2 \\
    \end{array}
\rb
\nabla G'=\cfrac{4\theta^2}{n^2\gamma^4}+\lb\cfrac{\omega}{\theta}(\nu_4-3)+2\rb\lb\cfrac{4\theta\lb \theta+\omega(\nu_4-3)+n\theta\rb^2}{\gamma^6n^3p}\rb.
\]
The result thus follows.
\end{proof}

  \subsection{Quasi-likelihood ratio test}

  Consider the Quasi-LRT statistic $\mathcal{L}_n$ in \eqref{quasiLRT}
based on the eigenvalues of $n-$dimensional matrix $\frac{1}{p}X'X$. Similarly with John's test statistic, it can be seen that, under the alternative hypothesis $H_1$, the $\mathcal{L}_n$ statistic can be represented as
\[
  \mathcal{L}_n=\frac{p}{n}\log\frac{\lb\frac{1}{n}\sum_{i=1}^n\tilde{l}_i\rb^n}{\prod_{i=1}^n\tilde{l}_i}
\]
where $\{\tilde{l}_i,~1\leq i\leq n\}$ are eigenvalues of $\frac{1}{p}Z'\Sigma_pZ$.  The main result of the power of the Quasi-LRT test is as follows.

\begin{theorem}\label{PowerThm2}
 With the same assumptions as in Theorem \ref{PowerThm1}, when $p/n\rightarrow\infty$, $n\rightarrow\infty$, $n^3/p=O(1)$,
    \[
   \mathcal{L}_n-\lb \frac{\theta}{2\gamma^2} n+\lb\frac{\theta^2}{2\gamma^4}-\frac{\theta\sqrt{\theta}}{3\gamma^3}\rb\frac{n^2}{p} \rb \xrightarrow{d} N\lb \frac{\theta}{2\gamma^2}+\frac{\omega}{2\gamma^2}(\nu_4-3) ,~\frac{\theta^2}{\gamma^4}\rb.
\]
\end{theorem}

Note that the theorem above reveals the limit distribution of the Quasi-LRT statistic under alternative hypothesis $H_1$. Nevertheless, if let $\Sigma_p=\sigma^2I_p$, then $\gamma=\sigma^2$, $\theta=\omega=\sigma^4$, Theorem \ref{PowerThm2} reduces to Theorem \ref{MainThm2}, which states the null distribution of the Quasi-LRT test statistic under $H_0$. Similarly, with the two limit distributions of QLRT statistic under $H_0$ and $H_1$, power of the test is derived as below.

\begin{proposition}\label{Prop:PowerQLRT}
  With the same assumptions as in Theorem \ref{PowerThm1}, when $p/n\rightarrow\infty,~ n\rightarrow\infty,~n^3/p=O(1)$,
  the power of QLRT $\beta_{\text{QLRT}}(H_1)$ is
   $$1-\Phi\lb\frac{\gamma^2}{\theta}Z_\alpha+\lb\frac{\gamma^2-\theta}{2\theta}\rb n+\lb\frac{\gamma^2}{6\theta}-\frac{\theta}{2\gamma^2}+\frac{\sqrt{\theta}}{3\gamma} \rb\frac{n^2}{p}+\lb \frac{\gamma^2(\nu_4-2)-\theta-\omega(\nu_4-3)}{2\theta}\rb\rb\rightarrow 1,$$
   where  $\alpha$ is the nominal test level, $Z_\alpha$, $\Phi(\cdot)$ are the alpha upper quantile and cdf of standard normal distribution respectively.
\end{proposition}

\noindent
For QLRT statistic $\mathcal{L}$, under $H_0$,
\[\mathcal{L}_n-\frac{n}{2}-\frac{n^2}{6p}\xrightarrow{d} N\lb \frac{\nu_4-2}{2},1\rb,
\]
under $H_1$,
  \[
   \mathcal{L}_n- \frac{\theta}{2\gamma^2} n-\lb\frac{\theta^2}{2\gamma^4}-\frac{\theta\sqrt{\theta}}{3\gamma^3}\rb\frac{n^2}{p} \xrightarrow{d} N\lb \frac{\theta}{2\gamma^2}+\frac{\omega}{2\gamma^2}(\nu_4-3) ,~\frac{\theta^2}{\gamma^4}\rb.
\]

\begin{align*}
  \beta_{\text{QLRT}}(H_1)&=P\lb\left. \mathcal{L}_n-\frac{n}{2}-\frac{n^2}{6p}-\frac{\nu_4-2}{2}>Z_\alpha\right| H_1 \rb\\
  &=P\left( \cfrac{\mathcal{L}_n- \frac{\theta}{2\gamma^2} n-\lb\frac{\theta^2}{2\gamma^4}-\frac{\theta\sqrt{\theta}}{3\gamma^3}\rb\frac{n^2}{p}-\lb \frac{\theta}{2\gamma^2}+\frac{\omega}{2\gamma^2}(\nu_4-3) \rb}{\frac{\theta}{\gamma^2}} \right.\\
  &  >\left.\cfrac{Z_\alpha+ \frac{n}{2}+\frac{n^2}{6p}+\frac{\nu_4-2}{2}- \frac{\theta}{2\gamma^2} n-\lb\frac{\theta^2}{2\gamma^4}-\frac{\theta\sqrt{\theta}}{3\gamma^3}\rb\frac{n^2}{p}-\lb \frac{\theta}{2\gamma^2}+\frac{\omega}{2\gamma^2}(\nu_4-3) \rb}{\frac{\theta}{\gamma^2}}  \right)\\
  &=1-\Phi\lb\frac{\gamma^2}{\theta}Z_\alpha+\lb\frac{\gamma^2-\theta}{2\theta}\rb n+\lb\frac{\gamma^2}{6\theta}-\frac{\theta}{2\gamma^2}+\frac{\sqrt{\theta}}{3\gamma} \rb\frac{n^2}{p}+\lb \frac{\gamma^2(\nu_4-2)-\theta-\omega(\nu_4-3)}{2\theta}\rb\rb,
  \end{align*}
since $\gamma^2\leq \theta$, Proposition \ref{Prop:PowerQLRT} follows.
\vspace{0.4cm}

\noindent
The proof of Theorem \ref{PowerThm2} is based on lemma \ref{Powerlem2}.


\begin{proof}
  Denote the eigenvalues of $n\times n$ matrix $\frac{1}{p}X'X=\frac{1}{p}Z'\Sigma_pZ$ in descending order by $\widetilde{l}_i(1\leq i\leq n)$, and eigenvalues of $n\times n$ matrix $\widetilde{A}=\sqrt{\cfrac{1}{n}}\lb \cfrac{1}{\sqrt{\tr (\Sigma_p^2)}} Z'\Sigma_p Z-\cfrac{\tr (\Sigma_p)}{\sqrt{\tr (\Sigma_p^2)}}  I_n\rb$ by $\widetilde{\lambda}_i(1\leq i\leq n)$. These eigenvalues are related as
  \[\sqrt{\frac{n\tr(\Sigma_p^2)}{p^2}}\widetilde{\lambda}_i+\frac{1}{p}\tr(\Sigma_p)=\widetilde{l}_i,~~1\leq i\leq n.\]
 We have, for the Quasi-LRT test statistic
  \begin{align*}
  \mathcal{L}_n  =&~\frac{p}{n}\log\left[\left.\lb\frac{1}{n}\sum_{i=1}^n\tilde{l}_i\rb^n\middle/\prod_{i=1}^n\tilde{l}_i\right.\right]\\
  =&~p\log\lb \gamma+\sqrt{\frac{\theta}{np}}\sum_{i=1}^n\widetilde{\lambda}_i\rb-\frac{p}{n}\sum_{i=1}^n\log\lb\gamma+\sqrt{\frac{n\theta}{p}}\widetilde{\lambda}_i\rb,
    \end{align*}
  Define the function
  $$G(u,v)=p\log \lb\gamma+\sqrt{\frac{\theta}{np}} u\rb-\sqrt{\frac pn}v,$$
  then the Quasi-LRT test statistic can be written as
  \[\mathcal{L}_n=G\lb u=\sum_{i=1}^n\widetilde{\lambda}_i,v=\sqrt{\frac{p}{n}}\sum_{i=1}^n\log\lb\gamma+\sqrt{\frac{n\theta}{p}}\widetilde{\lambda}_i\rb\rb.\]
  According to Lemma \ref{Powerlem2}, when $p/n\rightarrow\infty$, $n\rightarrow\infty$, $n^3/p=O(1)$,
 \[
  \left(
  \begin{array}{c}
    \sum_{i=1}^n\widetilde{\lambda}_i\\
    \sqrt{\frac pn}\sum_{i=1}^n\log \lb\gamma+\widetilde{\lambda}_i\sqrt{\frac {n\theta}{p}}\rb -\sqrt{pn}\log(\gamma)+\frac{\theta}{2\gamma^2}\sqrt{\frac{n^3}{p}}+\lb\lb \frac{\theta^2}{2\gamma^4}-\frac{\theta\sqrt{\theta}}{3\gamma^3}\rb\frac{n^2}{p}+\frac{\theta}{2\gamma^2}+\frac{\omega}{2\gamma^2}(\nu_4-3)\rb\sqrt{\frac np}
  \end{array}
  \right)\]\[=\xi_n+o_p(1),
  \]
  where
  \[\xi_n\sim N\left(\left(
  \begin{array}{c}
    0\\
    0
  \end{array}
  \right),~
  \left(
    \begin{array}{cc}
      \frac{\omega}{\theta}\lb\nu_4-3\rb+2 & \lb \frac{\omega}{\theta}\lb\nu_4-3\rb+2\rb\lb\frac{\sqrt{\theta}}{\gamma}+\frac{\theta\sqrt{\theta}}{\gamma^3}\frac np\rb \\
       \lb \frac{\omega}{\theta}\lb\nu_4-3\rb+2\rb\lb\frac{\sqrt{\theta}}{\gamma}+\frac{\theta\sqrt{\theta}}{\gamma^3}\frac np\rb  & \frac{\lb \frac{\omega}{\theta}\lb\nu_4-3\rb+2\rb\theta }{\gamma^2}+\frac{\lb \frac{2\omega}{\theta}\lb\nu_4-3\rb+5\rb\theta^2n}{\gamma^4p} \\
    \end{array}
  \right)
  \right)
  \]

By the Delta Method,
\[
\mathcal{L}_n-\left.G(u,v)\right\rvert_{u=0,v=\sqrt{pn}\log(\gamma)-\frac{\theta}{2\gamma^2}\sqrt{\frac{n^3}{p}}-\lb\lb \frac{\theta^2}{2\gamma^4}-\frac{\theta\sqrt{\theta}}{3\gamma^3}\rb\frac{n^2}{p}+\frac{\theta}{2\gamma^2}+\frac{\omega}{2\gamma^2}(\nu_4-3)\rb\sqrt{\frac np}}
\xrightarrow{d}\]
\[
N\lb 0,~
\nabla G\left(
    \begin{array}{cc}
      \frac{\omega}{\theta}\lb\nu_4-3\rb+2 & \lb \frac{\omega}{\theta}\lb\nu_4-3\rb+2\rb\lb\frac{\sqrt{\theta}}{\gamma}+\frac{\theta\sqrt{\theta}}{\gamma^3}\frac np\rb \\
       \lb \frac{\omega}{\theta}\lb\nu_4-3\rb+2\rb\lb\frac{\sqrt{\theta}}{\gamma}+\frac{\theta\sqrt{\theta}}{\gamma^3}\frac np\rb  & \frac{\lb \frac{\omega}{\theta}\lb\nu_4-3\rb+2\rb\theta }{\gamma^2}+\frac{\lb \frac{2\omega}{\theta}\lb\nu_4-3\rb+5\rb\theta^2n}{\gamma^4p} \\
    \end{array}
  \right)
\nabla G'
\rb,
\]
where $\nabla G=\left.\lb\cfrac{\partial U}{\partial u},\cfrac{\partial U}{\partial \nu}\rb\right\rvert_{u=0,v=\sqrt{pn}\log(\gamma)-\frac{\theta}{2\gamma^2}\sqrt{\frac{n^3}{p}}-\lb\lb \frac{\theta^2}{2\gamma^4}-\frac{\theta\sqrt{\theta}}{3\gamma^3}\rb\frac{n^2}{p}+\frac{\theta}{2\gamma^2}+\frac{\omega}{2\gamma^2}(\nu_4-3)\rb\sqrt{\frac np}}
$ is the corresponding gradient vector.
\\
\noindent
Then we have, for $(u,v)=\lb0, \sqrt{pn}\log(\gamma)-\frac{\theta}{2\gamma^2}\sqrt{\frac{n^3}{p}}-\lb\lb \frac{\theta^2}{2\gamma^4}-\frac{\theta\sqrt{\theta}}{3\gamma^3}\rb\frac{n^2}{p}+\frac{\theta}{2\gamma^2}+\frac{\omega}{2\gamma^2}(\nu_4-3)\rb\sqrt{\frac np} \rb$,
\[
G(u,v)=\frac{\theta}{2\gamma^2}n+\lb \frac{\theta^2}{2\gamma^4}-\frac{\theta\sqrt{\theta}}{3\gamma^3}\rb\frac{n^2}{p}+\frac{\theta}{2\gamma^2}+\frac{\omega}{2\gamma^2}(\nu_4-3),
\]
\noindent
and
\[
\nabla G ~Cov\lb u,v\rb
\nabla G'=~\frac{\theta^2}{\gamma^4}.
\]
The result thus follows.
  \end{proof}

   \subsection{Simulation Experiments}

   Empirical power of the two tests are shown in this section to testify the theoretical results presented in Proposition \ref{Prop:PowerJohn} and \ref{Prop:PowerQLRT}. Specifically, we consider two different scenarios to generate sample data:
  \begin{itemize}
    \item[(1)] $\{Z_j,~1\leq j\leq n\}$ i.i.d $p-$dimensional random vector generated from multivariate normal population $N_p({\bf 0},I_p)$, $\E z_{ij}^4=\nu_4=3$, $X_j=\Sigma_p^{1/2}Z_j$, $1\leq j\leq n$;
        \item[(2)] $\{z_{ij},~1\leq i\leq p,~1\leq j\leq n\}$ i.i.d follow $Gamma(4,2)-2$ distribution, then $\mathbb{E}z_{ij}=0$, $\mathbb{E}z^2_{ij}=1$, $\mathbb{E}z^4_{ij}=\nu_4=4.5$. $X_{p\times n}=\Sigma_p^{1/2}Z_{p\times n}$.
  \end{itemize}
  To cover multiple alternative hypothesis, $\Sigma_p$ is configured as a diagonal matrix with elements 0.5 and 1. The proportion of ``1" is $\delta$. The nominal test level is set as $\alpha=0.05$. $(p,n)=(2400,64)$ and empirical power are generated from 5000 replications. Theoretical values are displayed for comparison.

\begin{table}[!ht]
\caption{Empirical and Theoretical Power of two tests}\label{Tab:EmPower}
\begin{tabular}{|c|cc|cc|cc|cc|}
\hline
\multirow{3}{*}{$\bf \delta$} & \multicolumn{4}{c|}{\bf Gaussian} & \multicolumn{4}{c|}{\bf Non-Gaussian}\tabularnewline
\cline{2-9}
 & \multicolumn{2}{c|}{\bf John's test} & \multicolumn{2}{c|}{\bf QLRT} & \multicolumn{2}{c|}{\bf John's test} & \multicolumn{2}{c|}{\bf QLRT}\tabularnewline
\cline{2-9}
&\bf  Empirical &\bf  Theory &\bf  Empirical &\bf  Theory &\bf  Empirical &\bf  Theory &\bf  Empirical & \bf Theory\tabularnewline
 \hline

\bf0 & 0.046 & 0.050 & 0.049 & 0.050 & 0.051 & 0.050 & 0.052 & 0.050\tabularnewline
\bf0.1 & 0.738 & 0.745 & 0.727 & 0.759 & 0.736 & 0.746 & 0.727 & 0.761\tabularnewline
\bf0.2 & 0.958 & 0.953 & 0.954 & 0.959 & 0.950 & 0.954 & 0.951 & 0.960\tabularnewline
\bf0.3 & 0.984 & 0.979 & 0.982 & 0.982 & 0.981 & 0.979 & 0.981 & 0.982\tabularnewline
\bf0.4 & 0.978 & 0.976 & 0.978 & 0.980 & 0.978 & 0.976 & 0.978 & 0.980\tabularnewline
\bf0.5 & 0.958 & 0.953 & 0.958 & 0.959 & 0.951 & 0.954 & 0.950 & 0.960\tabularnewline
\hline
\end{tabular}
\end{table}

It can be seen from Table \ref{Tab:EmPower} that the empirical and theoretical power coincide with each other and both tests have very large power even when $\delta$ is small.

\section{Discussions and Auxiliary Results}\label{empsec}
In summary, we found in the considered ultra-dimension ($p\gg n$) situations, QLRT is the most recommended procedure regarding its maximal power for sphericity test. However, from the application perspective where the dimension $p$ and $n$ are explicitly known, it becomes very difficult to decide which asymptotic scheme to use, namely, `` $p$ fixed, $n\rightarrow\infty$", ``$p/n\rightarrow c\in(0,\infty), ~p,n\rightarrow\infty$", or ``$p/n\rightarrow\infty, ~ p,n\rightarrow\infty$" etc. Combining our study with the existing literature, we would like to recommend a  {\it dimension-proof~}  procedure like John's test or Chen's test, with a slight preference for John's test as it has a slightly higher power and an easier implementation.

We conclude the paper by mentioning some surprising consequence of the main results of the paper as follows.

\begin{corollary}\label{cor1}
  Assume $X=\{x_{ij}\}_{p\times n}$ are i.i.d. satisfying $\mathbb{E}(x_{ij})=0$, $\mathbb{E}(x_{ij}^2)=1$, $\mathbb{E}|x_{ij}|^4=\nu_4<\infty$, then when $n/p\rightarrow\infty$, $n,p \rightarrow\infty$,
  \begin{equation*}
   -\frac 2p\log L_n-\frac{p}{2}-\frac{p^2}{6n}-\frac{\nu_4-2}{2} \xrightarrow{d} N\lb 0,1\rb.
  \end{equation*}
  where $-\frac 2p\log L_n=\frac np\log\lb\frac{\overline{l}^p}{\prod_{i=1}^p l_i}\rb$, $\{l_i\}_{1\leq i\leq p}$ are the eigenvalues of $p-$dimensional sample covariance matrix $\frac1nXX'$.
\end{corollary}
\noindent
Note that if we fix $p$ while let $n\rightarrow \infty$, under normality assumption, the Corollary \ref{cor1} reduces to
  \[-\frac 2p\log L_n-\frac{p+1}{2}\xrightarrow{d} N\lb 0,1\rb,\]
which is consistent with the classic LRT asymptotic, i.e.
$ -2\log L_n \xrightarrow{d} \chi^2_{\frac{1}{2}p(p+1)-1}.$

\begin{corollary}\label{cor2}
    Assume $X=\{x_{ij}\}_{p\times n}$ are i.i.d. satisfying $\mathbb{E}(x_{ij})=0$, $\mathbb{E}(x_{ij}^2)=1$, $\mathbb{E}|x_{ij}|^4=\nu_4<\infty$, then when $n/p\rightarrow\infty$, $n,p \rightarrow\infty$,
      \[nU-p \xrightarrow{d} N(\nu_4-2,4).\]
\end{corollary}

\begin{proof}
Interchanging the role of $n$ and $p$ in Theorem \ref{MainThm}, keeping all other assumptions unchanged, it can be seen that, when $n/p\rightarrow\infty$, $n,p \rightarrow\infty$,
\[p\widetilde{U}-n \xrightarrow{d} N(\nu_4-2,4),\]
where
$$\widetilde{U}=\cfrac{n^{-1}\sum_{i=1}^n\tilde{l}_i^2}{\lb\frac 1n\sum_{i=1}^n\tilde{l}_i\rb^2}-1,$$
  $\tilde{l}_i(1\leq i\leq n)$ are eigenvalues of $n\times n$ matrix $\frac{1}{p}X'X$, $l_i$ are eigenvalues of $\frac 1n XX'$, then
  \begin{align*}
    p\widetilde{U}-n & =\cfrac{\frac pn\sum_{i=1}^n\tilde{l}_i^2}{\lb\frac 1n\sum_{i=1}^n\tilde{l}_i\rb^2}-p-n\\
    &= \cfrac{\frac np\sum_{i=1}^p l_i^2}{\overline{l}^2}-n-p= nU-p.
  \end{align*}
\end{proof}
\noindent
 Henceforth, the {\it dimension-proof} property of John's test statistic, i.e. regardless of normality, under any $(n,p)$-asymptotic, $n/p\rightarrow[0,\infty]$, has been completely testified.

\newpage
\appendix

\section{Technique Lemmas and additional proofs}\label{lemsec}

\begin{lemma}\label{techlem}
 In the central limit theorem of linear functions of eigenvalues of the re-normalized sample covariance matrix $A$ when the dimension $p$ is much larger than the sample size $n$ derived by \citet{ChenPan13}, Let $\mathscr{S}$ denote any open region on the complex plane including $[-2,2]$, the support of the semicircle law $F(x)$, we denote the Stieltjes transform of the semicircle law $F$ by $m(z)$. Let $\mathscr{M}$ be the set of functions which are analytic on $\mathscr{S}$, for any analytic function $f\in \mathcal{M}$, the mean correction term is defined as
\[\cfrac{n}{2\pi i}\oint_{|m|=\rho}f\lb -m-m^{-1}\rb \chi_n^{\rm{Calib}}(m)\cfrac{1-m^2}{m^2} \operatorname{d}\!m.\]
Define functions $f_1(x)=x^2, ~f_2(x)=x, ~ f_3(x)=\frac{p}{n}\log(1+\sqrt{\frac{n}{p}}x)$, then the mean correction term in equation \eqref{MeanTerm} for these functions are as follows:
  \[\cfrac{n}{2\pi i}\oint_{|m|=\rho}f_1\lb -m-m^{-1}\rb \chi_n^{\rm{Calib}}(m)\cfrac{1-m^2}{m^2} \operatorname{d}\!m=\nu_4-2,\]
  \[\cfrac{n}{2\pi i}\oint_{|m|=\rho}f_2\lb -m-m^{-1}\rb \chi_n^{\rm{Calib}}(m)\cfrac{1-m^2}{m^2} \operatorname{d}\!m=0,\]
  \[\cfrac{n}{2\pi i}\oint_{|m|=\rho}f_3\lb -m-m^{-1}\rb \chi_n^{\rm{Calib}}(m)\cfrac{1-m^2}{m^2} \operatorname{d}\!m=-\frac{\nu_4-2}{2}+\frac{n^2}{3p}.\]
\end{lemma}

\begin{proof}
Since
  \[\chi_n^{\rm{Calib}}(m)\triangleq \cfrac{-\mathcal{B}+\sqrt{\mathcal{B}^2-4\mathcal{A}\mathcal{C}^{\rm{Calib}}}}{2\mathcal{A}},~\mathcal{A}=m-\sqrt{\cfrac{n}{p}}(1+m^2),\]
\[\mathcal{B}=m^2-1-\cfrac{n}{p}m(1+2m^2),~\mathcal{C}^{\rm{Calib}}=\cfrac{m^3}{n}\left[\nu_4-2+\cfrac{m^2}{1-m^2}-2(\nu_4-1)m\sqrt{\cfrac{n}{p}}\right]-\sqrt{\cfrac{n}{p}}m^4,\]
the integral's contour is taken as $|m|=\rho$ with $\rho<1$.

\vspace{0.5 cm}
For $f_1(x)=x^2$, choose $\rho< \sqrt{\frac np}<\sqrt{\frac pn}$,
\begin{align*}
~&\cfrac{n}{2\pi i}\oint_{|m|=\rho}f_1\lb -m-m^{-1}\rb \chi_n^{\rm{Calib}}(m)\cfrac{1-m^2}{m^2} \operatorname{d}\!m\\
  =& ~\frac{n}{2\pi i}\oint_{|m|=\rho}\cfrac{(1-m^4)(1+m^2)}{m^4}\cdot\cfrac{X}{\lb1-\sqrt{\frac np}m\rb}\cdot\frac{1}{m-\sqrt{\frac np}} \operatorname{d}\!m\\
  ~&~\lb \mbox{denote } X:= \frac12\lb-\mathcal{B}+\sqrt{\mathcal{B}^2-4\mathcal{A}\mathcal{C}^{\rm{Calib}}}\rb\rb\\
  =&~\frac{n}{2\pi i}\oint_{|m|=\rho} \frac{1+m^2}{m^4}\cdot\cfrac{X}{\lb1-\sqrt{\frac np}m\rb}\cdot\frac{1}{m-\sqrt{\frac np}} \operatorname{d}\!m\\
  ~&~\lb \mbox{Cauchy's Residue Theorem}\rb\\
  =&~\frac{1}{3!}\left.{\operatorname{d}}^{(3)}\lb \cfrac{n X}{\lb1-\sqrt{\frac np}m\rb}\cdot\frac{1}{m-\sqrt{\frac np}}\rb\middle/ \operatorname{d} m^3\right|_{m=0}+\left. \operatorname{d} \lb \cfrac{n X}{\lb1-\sqrt{\frac np}m\rb}\cdot\frac{1}{m-\sqrt{\frac np}}\rb\middle/  \operatorname{d}\!m\right|_{m=0}\\
  =&~\nu_4-2.
\end{align*}

For $f_2(x)=x$, choose $\rho< \sqrt{\frac np}<\sqrt{\frac pn}$,
\begin{align*}
  ~&\cfrac{n}{2\pi i}\oint_{|m|=\rho}f_2\lb -m-m^{-1}\rb \chi_n^{\rm{Calib}}(m)\cfrac{1-m^2}{m^2} \d m\\
  =&~ \cfrac{n}{2\pi i}\oint_{|m|=\rho} (-m-m^{-1})\cdot \cfrac{X}{\lb1-\sqrt{\frac np}m\rb}\cdot\frac{1}{m-\sqrt{\frac np}}\cdot\frac{1-m^2}{m^2}\d m\\
  =&~ - \cfrac{n}{2\pi i}\oint_{|m|=\rho}\frac{1}{m^3}\cdot\cfrac{X}{\lb1-\sqrt{\frac np}m\rb}\cdot\frac{1}{m-\sqrt{\frac np}}\d m\\
  =&~~-\frac{1}{2!}\left.{\operatorname{d}}^{(2)}\lb \cfrac{n X}{\lb1-\sqrt{\frac np}m\rb}\cdot\frac{1}{m-\sqrt{\frac np}}\rb\middle/ \operatorname{d} m^2\right|_{m=0}=~0.
\end{align*}
For $f_3(x)=\frac{p}{n}\log(1+\sqrt{\frac{n}{p}}x)$, choose $\rho< \sqrt{\frac np}<\sqrt{\frac pn}$,
\begin{align*}
  ~&\cfrac{n}{2\pi i}\oint_{|m|=\rho}f_3\lb -m-m^{-1}\rb \chi_n^{\rm{Calib}}(m)\cfrac{1-m^2}{m^2} \d m\\
  =&~\cfrac{p}{2\pi i}\oint_{|m|=\rho}\log \lb 1-\sqrt{\frac np}(m+m^{-1})\rb \chi_n^{\rm{Calib}}(m)\cfrac{1-m^2}{m^2} \d m\\
  =&~\cfrac{p}{2\pi i}\oint_{|m|=\rho}\sum_{k=1}^\infty\left[-\frac 1k \lb\sqrt{\frac np}(m+m^{-1})\rb^k\right]\cdot \cfrac{X}{\lb1-\sqrt{\frac np}m\rb}\cdot\frac{1}{m-\sqrt{\frac np}}\cdot\frac{1-m^2}{m^2}\d m\\
  =&~-\cfrac{1}{2\pi i}\oint_{|m|=\rho}\left[\sqrt{np}\cdot\frac{1-m^4}{m^3}+\frac n2\cdot\frac{(1-m^4)(1+m^2)}{m^4}+\frac n3 \cdot \sqrt{\frac np}\cdot\frac{(1-m^4)(1+m^2)^2}{m^5}\right.\\
  ~&~~~\left.+\frac{n^2}{4p}\cdot\frac {(1-m^4)(1+m^2)^3}{m^6}\right]\cdot \cfrac{X}{\lb1-\sqrt{\frac np}m\rb}\cdot\frac{1}{m-\sqrt{\frac np}} \d m + o\lb\frac{n^2}{p}\rb  \\
  =&~-\cfrac{1}{2\pi i}\oint_{|m|=\rho}\left[\sqrt{np}\cdot\frac{1}{m^3}+\frac n2\cdot\frac{1+m^2}{m^4}+\frac n3 \cdot \sqrt{\frac np}\cdot\frac{1+2m^2}{m^5}\right.\\
  ~&~~~\left.+\frac{n^2}{4p}\cdot\frac {1+3m^2+2m^4}{m^6}\right]\cdot \cfrac{X}{\lb1-\sqrt{\frac np}m\rb}\cdot\frac{1}{m-\sqrt{\frac np}} \d m + o\lb\frac{n^2}{p}\rb
\end{align*}

According to Cauchy's residue theorem, we have
\begin{align*}
  ~&~-\cfrac{1}{2\pi i}\oint_{|m|=\rho}\sqrt{np}\cdot\frac{1}{m^3}\cdot \cfrac{X}{\lb1-\sqrt{\frac np}m\rb}\cdot\frac{1}{m-\sqrt{\frac np}} \d m\\
  =&~-\frac{1}{2!}\left.{\operatorname{d}}^{(2)}\lb \cfrac{\sqrt{np} X}{\lb1-\sqrt{\frac np}m\rb}\cdot\frac{1}{m-\sqrt{\frac np}}\rb\middle/ \operatorname{d} m^2\right|_{m=0}=~0,
\end{align*}
similarly,
\begin{align*}
  ~&-\cfrac{1}{2\pi i}\oint_{|m|=\rho}\left[\sqrt{np}\cdot\frac{1}{m^3}+\frac n2\cdot\frac{1+m^2}{m^4}+\frac n3 \cdot \sqrt{\frac np}\cdot\frac{1+2m^2}{m^5}\right.\\
  ~&~~~\left.+\frac{n^2}{4p}\cdot\frac {1+3m^2+2m^4}{m^6}\right]\cdot \cfrac{X}{\lb1-\sqrt{\frac np}m\rb}\cdot\frac{1}{m-\sqrt{\frac np}} \d m + o\lb\frac{n^2}{p}\rb\\
  =&~0-\frac{\nu_4-2}{2}-0+\frac{p^2+\nu_4-2}{3n}-0-\frac{(2\nu_4-3)n}{p}-\frac{(\nu_4-2)n^2}{4p^2}-\frac{3(\nu_4-2)n}{4p}-0+ o\lb\frac{n^2}{p}\rb\\
  =&~-\frac{\nu_4-2}{2}+\frac{n^2}{3p}+ o\lb\frac{n^2}{p}\rb.
\end{align*}

\end{proof}



\noindent
\textbf{Proof of Lemma \ref{lemma} :}

\begin{proof}
  According to Theorem \ref{CLT}, define function $f_1(x)=x^2$, then
\begin{align*}
   G_n(f_1)&=n\int_{-\infty}^{+\infty}f_1(x)d\lb F^A(x)-F(x)\rb \\
   &=\sum_{i=1}^n\lambda_i^2-n\int_{-\infty}^{+\infty} x^2\cdot \frac{1}{2\pi}\sqrt{4-x^2}\d x\\
   &=\sum_{i=1}^n\lambda_i^2-n,
\end{align*}
where $F^A$ is ESD of $A=\cfrac{1}{\sqrt{np}}\lb X'X-pI_n\rb$ and $F$ represents the semicircular law.
The mean correction term for $f_1(x)=x^2$ is, according to Lemma \ref{techlem},
  \[\cfrac{n}{2\pi i}\oint_{|m|=\rho}f_1\lb -m-m^{-1}\rb \chi_n^{\rm{Calib}}(m)\cfrac{1-m^2}{m^2} \operatorname{d}\!m=\nu_4-2,\]
As for the mean function and covariance function of the Gaussian limit $Y(f_1)$,
since
\[\Phi_1(f_1)=\cfrac{1}{2\pi}\int_{-\pi}^{\pi}4\cos^3\theta \d\theta=0,\]
\[\Phi_2(f_1)=\cfrac{1}{2\pi}\int_{-\pi}^{\pi}4\cos^2\theta \cos 2\theta \d\theta=\cfrac{1}{2\pi}\int_{-\pi}^{\pi}(\cos 4\theta+1+2\cos 2\theta)\d\theta=1,\]
\begin{align*}
 \Phi_k(f_1)&=\cfrac{1}{2\pi}\int_{-\pi}^{\pi}4\cos^2\theta \cos k\theta \d\theta=\cfrac{1}{2\pi}\int_{-\pi}^{\pi}2(\cos 2\theta+1)\cos k\theta \d\theta\\
 &=\cfrac{1}{2\pi}\int_{-\pi}^{\pi}(\cos(k-2)\theta+\cos(k+2)\theta+2\cos k\theta)\d\theta=0,~~ \mbox{ for }k\geq 3,
\end{align*}
therefore $Var(Y(f_1))=4$, in addition, $\mathbb{E}(Y(f_1))=0$,
Conclusively, we have, when $p/n\rightarrow\infty$, $n\rightarrow\infty$,
\[\sum_{i=1}^n\lambda_i^2-n-(\nu_4-2)\xrightarrow{d} N(0,4).\]

Similarly, if we define function $f_2=x$, then
\begin{align*}
   G_n(f_2)&=n\int_{-\infty}^{+\infty}f_2(x) d \lb F^A(x)-F(x)\rb \\
   &=\sum_{i=1}^n\lambda_i-n\int_{-\infty}^{+\infty} x\cdot \frac{1}{2\pi}\sqrt{4-x^2}dx=\sum_{i=1}^n\lambda_i.
\end{align*}
The mean correction term for $f_2(x)=x$ is, according to Lemma \ref{techlem},
  \[\cfrac{n}{2\pi i}\oint_{|m|=\rho}f_2\lb -m-m^{-1}\rb \chi_n^{\rm{Calib}}(m)\cfrac{1-m^2}{m^2} \operatorname{d}\!m=0,\]
As for the mean function and covariance function of the Gaussian limit $Y(f_2)$,
since
\[\Phi_0(f_2)=\cfrac{1}{2\pi}\int_{-\pi}^{\pi}2\cos\theta \d\theta=0,\]
\[\Phi_1(f_2)=\cfrac{1}{2\pi}\int_{-\pi}^{\pi}2\cos^2\theta \d\theta=1,\]
\[\Phi_2(f_2)=\cfrac{1}{2\pi}\int_{-\pi}^{\pi}2\cos\theta \cos 2\theta \d\theta=\cfrac{1}{2\pi}\int_{-\pi}^{\pi}(\cos 3\theta+\cos \theta)\d\theta=0,\]
\begin{align*}
 \Phi_k(f_2)&=\cfrac{1}{2\pi}\int_{-\pi}^{\pi}2\cos\theta \cos k\theta \d\theta\\
 &=\cfrac{1}{2\pi}\int_{-\pi}^{\pi}(\cos(k+1)\theta+\cos(k-1)\theta)\d\theta=0 \mbox{ for }k\geq 3,
\end{align*}
therefore
\begin{align*}
  Var(G_n(f_2))&=(\nu_4-3)\Phi_1(f_2)\Phi_1(f_2)+2\sum_{k=1}^{\infty}k\Phi_k(f_2)\Phi_k(f_2)\\
  &=\nu_4-1,
\end{align*}
in addition, $\mathbb{E}(Y(f_2))=0$. In conclusion, we have, when $p/n\rightarrow\infty$, $n\rightarrow\infty$,
\[\sum_{i=1}^n\lambda_i^2-n-(\nu_4-2)\xrightarrow{d} N(0,4),\]
\[\sum_{i=1}^n\lambda_i\xrightarrow{d} N(0,\nu_4-1).\]
Now consider the covariance between $G_n(f_1)$ and $G_n(f_2)$, then
\[Cov(G_n(f_1),G_n(f_2))=(\nu_4-3)\Phi_1(f_1)\Phi_1(f_2)+2\sum_{k=1}^{\infty}k\Phi_k(f_1)\Phi_k(f_2)=0.\]
Consequently, when $p/n\rightarrow\infty$, $n\rightarrow\infty$,
\[
\left(
\begin{array}{c}
 \sum_{i=1}^n\lambda_i^2-n-(\nu_4-2)\\
  \sum_{i=1}^n\lambda_i
\end{array}
\right)\xrightarrow{d} N\left(
\lb
\begin{array}{c}
  0\\
  0
\end{array}
\rb,~\lb
\begin{array}{cc}
 4&0\\
 0&\nu_4-1
\end{array}
\rb
\right),
\]
\end{proof}


\noindent
\textbf{Proof of Lemma \ref{LRTlem} :}

\begin{proof}
   According to Theorem \ref{CLT}, define function $f_3(x)= \frac pn \log \lb1+\sqrt{\frac np}x\rb$, then
\begin{align*}
   G_n(f_3)&=n\int_{-\infty}^{+\infty}f_3(x)\operatorname{d}\! \lb F^A(x)-F(x)\rb \\
   &=\frac pn\sum_{i=1}^n\log \lb1+\lambda_i\sqrt{\frac np}\rb-n\int_{-2}^{2} \frac pn \log \lb1+\sqrt{\frac np}x\rb\cdot F(x)\d x
\end{align*}
where $F^A$ is ESD of $A=\cfrac{1}{\sqrt{np}}\lb X'X-pI_n\rb$ and $F$ represents the semicircular law.
\begin{align*}
  n\int_{-2}^{2} \frac pn \log \lb1+\sqrt{\frac np}x\rb\cdot F(x)\d x
  =&~p\int_{-2}^{2} \log \lb1+\sqrt{\frac np}x\rb\cdot \frac{1}{2\pi}\sqrt{4-x^2}\d x\\
  =&-\frac n2\sum_{k=0}^{\infty}\frac{(2k+1)!!}{2^{k-1}(k+1)^2(k+2)}\cdot\lb\frac{4n}{p}\rb^k\\
  =&~-\frac n2-\frac{n^2}{2p}+o\lb\frac{n^2}{2p}\rb
\end{align*}

The mean correction term for $f_3(x)=\frac pn \log \lb1+\sqrt{\frac np}x\rb$ is, according to Lemma \ref{techlem},
  \[\cfrac{n}{2\pi i}\oint_{|m|=\rho}f_3\lb -m-m^{-1}\rb \chi_n^{\rm{Calib}}(m)\cfrac{1-m^2}{m^2} \operatorname{d}\!m=-\frac{\nu_4-2}{2}+\frac{n^2}{3p},\]
As for the mean function and covariance function of the Gaussian limit $Y(f_3)$,
since $\log(1+x)=\sum_{n=1}^{\infty}(-1)^{n+1}\frac{x^n}{n},$
\begin{align*}
  \Phi_1(f_3)&=\frac{1}{2\pi}\int_{-\pi}^{\pi}f_3(2\cos\theta)\cdot \cos\theta\d\theta\\
  ~&=~\frac{1}{2\pi}\int_{-\pi}^{\pi}\sqrt{\frac pn}\cdot 2\cos \theta\cdot\cos \theta\d\theta-\frac{1}{2\pi}\int_{-\pi}^{\pi}\frac 12\cdot (2\cos \theta)^2\cdot\cos \theta\d\theta\\
  ~&+\frac{1}{2\pi}\int_{-\pi}^{\pi}\frac13\sqrt{\frac np}\cdot(2\cos\theta)^3\cdot\cos\theta\d\theta+o\lb\sqrt{\frac np}\rb\\
  ~&=\sqrt{\frac pn}+\sqrt{\frac np}+o\lb\sqrt{\frac np}\rb,
\end{align*}
for $k\geq 2$,
\begin{align*}
 \Phi_k(f_3)&=\cfrac{1}{2\pi}\int_{-\pi}^{\pi}f_3(2\cos \theta) \cos k\theta \d\theta\\
 ~&=~\frac{1}{2\pi}\int_{-\pi}^{\pi}\sqrt{\frac pn}\cdot 2\cos \theta\cdot\cos k\theta\d\theta-\frac{1}{2\pi}\int_{-\pi}^{\pi}\frac 12\cdot (2\cos \theta)^2\cdot\cos k\theta\d\theta\\
  ~&+\frac{1}{2\pi}\int_{-\pi}^{\pi}\frac13\sqrt{\frac np}\cdot(2\cos\theta)^3\cdot\cos k\theta\d\theta+o\lb\sqrt{\frac np}\rb\\
  ~&=o\lb\sqrt{\frac np}\rb+\begin{cases}
    -\frac{1}{2} & \text{$k=2$}\\
    \frac{1}{3}\sqrt{\frac np} &\text{$k=3$}\\
    0 &\text{$k\geq 4$}
  \end{cases}
\end{align*}
therefore
\begin{align*}
  Var(G_n(f_3))&=(\nu_4-3)\Phi_1(f_3)\Phi_1(f_3)+2\sum_{k=1}^{\infty}k\Phi_k(f_3)\Phi_k(f_3)\\
  &=(\nu_4-1)\lb\sqrt{\frac pn}+\sqrt{\frac np}\rb^2+ 2\cdot 2\cdot \lb-\frac12\rb^2+2\cdot3 \cdot\frac{1}{9}\cdot \frac{n}{p}\\
 &=(\nu_4-1)\cdot\frac pn+2\nu_4-1+\frac np(\nu_4-\frac13),
\end{align*}
in addition, $\mathbb{E}(Y(f_3))=0$,
Conclusively, we have, when $p/n\rightarrow\infty$, $n\rightarrow\infty$,
\[\frac pn\sum_{i=1}^n\log\lb 1+\sqrt{\frac np}\lambda_i\rb+\frac n2+\frac{n^2}{6p}+\frac{\nu_4-2}{2}\xrightarrow{d} N\lb 0,\frac pn(\nu_4-1)+2\nu_4-1+\frac np\lb\nu_4-\frac13\rb\rb.\]

If we define function $f_2=x$, it has been proved in Lemma \ref{lemma} that,
\[
   G_n(f_2)=n\int_{-\infty}^{+\infty}f_2(x) d \lb F^A(x)-F(x)\rb=\sum_{i=1}^n\lambda_i.
\]
The mean correction term for $f_2(x)=x$ is, according to Lemma \ref{techlem},
  \[\cfrac{n}{2\pi i}\oint_{|m|=\rho}f_2\lb -m-m^{-1}\rb \chi_n^{\rm{Calib}}(m)\cfrac{1-m^2}{m^2} \operatorname{d}\!m=0,\]
As for the mean function and covariance function of the Gaussian limit $Y(f_2)$,
\begin{align*}
\Phi_0(f_2)&=\cfrac{1}{2\pi}\int_{-\pi}^{\pi}2\cos\theta \d\theta=0,\\
\Phi_1(f_2)&=\cfrac{1}{2\pi}\int_{-\pi}^{\pi}2\cos^2\theta \d\theta=1,\\
\Phi_2(f_2)&=\cfrac{1}{2\pi}\int_{-\pi}^{\pi}2\cos\theta \cos 2\theta \d\theta=0,\\
 \Phi_k(f_2)&=\cfrac{1}{2\pi}\int_{-\pi}^{\pi}2\cos\theta \cos k\theta \d\theta=0 \mbox{ for }k\geq 3,
\end{align*}

\[
  Var(G_n(f_2))=(\nu_4-3)\Phi_1(f_2)\Phi_1(f_2)+2\sum_{k=1}^{\infty}k\Phi_k(f_2)\Phi_k(f_2)=\nu_4-1,
\]
in addition, $\mathbb{E}(Y(f_2))=0$. In conclusion, we have, when $p/n\rightarrow\infty$, $n\rightarrow\infty$,
\[\sum_{i=1}^n\lambda_i\xrightarrow{d} N(0,\nu_4-1).\]
Now consider the covariance between $G_n(f_3)$ and $G_n(f_2)$, then
\begin{align*}
  Cov(G_n(f_3),G_n(f_2))&=(\nu_4-3)\Phi_1(f_3)\Phi_1(f_2)+2\sum_{k=1}^{\infty}k\Phi_k(f_3)\Phi_k(f_2)=(\nu_4-1)\lb\sqrt{\frac pn}+\sqrt{\frac np}\rb.
\end{align*}

Consequently result follows.
\end{proof}

{\bf Proof of Lemma \ref{Powerlem1}:}
\par
\vspace{0.2cm}

\begin{proof}
According to Theorem \ref{PowerCLT}, define function $f_{1}\left(x\right)=x^{2}$,
then

\begin{align*}
G_{n}\left(f_{1}\right) & =  n\int_{-\infty}^{+\infty}f_{1}\left(x\right) d\left(F^{\widetilde{A}}\left(x\right)-F\left(x\right)\right)-\sqrt{\frac{n^{3}}{p}}\Phi_{3}\left(f_{1}\right)\\
 & =  \sum_{i=1}^{n}\widetilde{\lambda}_{i}^{2}-n\int_{-2}^{2}\frac{x^{2}}{2\pi}\sqrt{4-x^{2}}\d x-\sqrt{\frac{n^{3}}{p}}\Phi_{3}\left(f_{1}\right)\\
 & =  \sum_{i=1}^{n}\widetilde{\lambda}_{i}^{2}-n-\sqrt{\frac{n^{3}}{p}}\Phi_{3}\left(f_{1}\right),\end{align*}

where $F^{\widetilde{A}}$ is the ESD of $\widetilde{A}=\sqrt{\frac{1}{n}}\left(\frac{1}{\sqrt{\tr\left(\Sigma_{p}^{2}\right)}}Z'\Sigma_{p}Z-\frac{\tr\left(\Sigma_{p}\right)}{\sqrt{\tr\left(\Sigma_{p}^{2}\right)}}I_{n}\right)$
and $F$ represents the semicircular law.

As for the mean function and covariance function of the Gaussian limit
$Y\left(f_{1}\right)$, since

\[
\Phi_{0}\left(f_{1}\right)=\frac{1}{2\pi}\int_{-\pi}^{\pi}f_{1}\left(2\cos\left(\theta\right)\right)\d\theta=2,\]

\[
\Phi_{k}\left(f_{1}\right)=\frac{1}{2\pi}\int_{-\pi}^{\pi}f_{1}\left(2\cos\left(\theta\right)\right)\cos\left(k\theta\right)\d\theta=\begin{cases}
0, & k=1,\\
1, & k=2,\\
0, & k\geq3,\end{cases}\]

$
G_{n}\left(f_{1}\right)=\sum_{i=1}^{n}\widetilde{\lambda}_{i}^{2}-n,$

\begin{align*}
\E\left(Y\left(f_{1}\right)\right) & = \frac{1}{4}\left(f_{1}\left(2\right)+f_{1}\left(-2\right)\right)-\frac{1}{2}\Phi_{0}\left(f_{1}\right)+\frac{\omega}{\theta}\left(\nu_{4}-3\right)\Phi_{2}\left(f_{1}\right)\\
 & =  2-1+\frac{\omega}{\theta}\left(\nu_{4}-3\right)=\frac{\omega}{\theta}\left(\nu_{4}-3\right)+1,\end{align*}

\[
Var\left(Y\left(f_{1}\right)\right)=\frac{\omega}{\theta}\left(\nu_{4}-3\right)\Phi_{1}^{2}\left(f_{1}\right)+2\sum_{k=1}^{\infty}k\Phi_{k}^{2}\left(f_{1}\right)=4.\]

Similarly, if we define function $f_{2}\left(x\right)=x$, then
\begin{align*}
G_{n}\left(f_{2}\right) & = n\int_{-\infty}^{+\infty}f_{2}\left(x\right)d\left(F^{\widetilde{A}}\left(x\right)-F\left(x\right)\right)-\sqrt{\frac{n^{3}}{p}}\Phi_{3}\left(f_{2}\right)\\
 & =  \sum_{i=1}^{n}\widetilde{\lambda}_{i}-n\int_{-2}^{2}\frac{x}{2\pi}\sqrt{4-x^{2}}\d x-\sqrt{\frac{n^{3}}{p}}\Phi_{3}\left(f_{2}\right)\\
 & =  \sum_{i=1}^{n}\widetilde{\lambda}_{i}-\sqrt{\frac{n^{3}}{p}}\Phi_{3}\left(f_{2}\right),\end{align*}

As for the mean function and covariance function of the Gaussian limit
$Y\left(f_{2}\right)$, since

\[
\Phi_{0}\left(f_{2}\right)=\frac{1}{2\pi}\int_{-\pi}^{\pi}f_{2}\left(2\cos\left(\theta\right)\right)\d\theta=0,\]

\[
\Phi_{k}\left(f_{2}\right)=\frac{1}{2\pi}\int_{-\pi}^{\pi}f_{2}\left(2\cos\left(\theta\right)\right)\cos\left(k\theta\right)\d\theta=\begin{cases}
1, & k=1,\\
0, & k\geq2,\end{cases}\]

$
G_{n}\left(f_{2}\right)=\sum_{i=1}^{n}\widetilde{\lambda}_{i},$

\begin{align*}
\E\left(Y\left(f_{2}\right)\right) & =  \frac{1}{4}\left(f_{2}\left(2\right)+f_{2}\left(-2\right)\right)-\frac{1}{2}\Phi_{0}\left(f_{2}\right)+\frac{\omega}{\theta}\left(\nu_{4}-3\right)\Phi_{2}\left(f_{2}\right)\\
 & =  -\frac{1}{2}\Phi_{0}\left(f_{2}\right)+\frac{\omega}{\theta}\left(\nu_{4}-3\right)\Phi_{2}\left(f_{2}\right)=0,
 \end{align*}

\[
Var\left(Y\left(f_{2}\right)\right)=\frac{\omega}{\theta}\left(\nu_{4}-3\right)+2,\]

\[
Cov\left(Y\left(f_{1}\right),Y\left(f_{2}\right)\right)=\frac{\omega}{\theta}\left(\nu_{4}-3\right)\Phi_{1}\left(f_{1}\right)\Phi_{1}\left(f_{2}\right)+2\sum_{k=1}^{\infty}k\Phi_{k}\left(f_{1}\right)\Phi_{k}\left(f_{2}\right)=0,\]

therefore

\[
\left(\begin{array}{c}
\sum_{i=1}^{n}\widetilde{\lambda}_{i}^{2}-n-\left(\frac{\omega}{\theta}\left(\nu_{4}-3\right)+1\right)\\
\sum_{i=1}^{n}\widetilde{\lambda}_{i}\end{array}\right)\xrightarrow{d}N_{2}\left(\left(\begin{array}{c}
0\\
0\end{array}\right),\left(\begin{array}{cc}
4 & 0\\
0 & \frac{\omega}{\theta}\left(\nu_{4}-3\right)+2\end{array}\right)\right),\]

as $n\rightarrow\infty,\quad p\rightarrow\infty,\quad p/n^{3}=O\left(1\right)$.
\end{proof}

\par
\vspace{0.5cm}

{\bf Proof of Lemma \ref{Powerlem2}:}
\par
\vspace{0.5cm}
\begin{proof}
According to Theorem \ref{PowerCLT}, define function $f_{3}\left(x\right)=\frac{p}{n}\log\left(\gamma+\sqrt{\frac{n\theta}{p}}x\right)$,
then

\begin{align*}
G_{n}\left(f_{3}\right) & =  n\int_{-\infty}^{+\infty}f_{3}\left(x\right)d\left(F^{\widetilde{A}}\left(x\right)-F\left(x\right)\right)-\sqrt{\frac{n^{3}}{p}}\Phi_{3}\left(f_{3}\right)\\
 & = \sum_{i=1}^{n}\frac{p}{n}\log\left(\gamma+\sqrt{\frac{n\theta}{p}}\widetilde{\lambda}_{i}\right)-n\int_{-2}^{2}\frac{p}{n}\log\left(\gamma+\sqrt{\frac{n\theta}{p}}x\right)\frac{1}{2\pi}\sqrt{4-x^{2}}\d x-\sqrt{\frac{n^{3}}{p}}\Phi_{3}\left(f_{3}\right),
 \end{align*}

\begin{align*}
&~n\int_{-2}^{2}\frac{p}{n}\log\left(\gamma+\sqrt{\frac{n\theta}{p}}x\right)\frac{1}{2\pi}\sqrt{4-x^{2}}\d x \\
 & = p\int_{-2}^{2}\left(\log\gamma+\sum_{k=1}^{\infty}\left(-1\right)^{k+1}\frac{\left(\frac{\sqrt{\theta}}{\gamma}\sqrt{\frac{n}{p}}x\right)^{k}}{k}\right)\frac{1}{2\pi}\sqrt{4-x^{2}}\d x\\
 & =  p\log\gamma-p\int_{-2}^{2}\frac{\theta}{2\gamma^{2}}\frac{n}{p}x^{2}\frac{1}{2\pi}\sqrt{4-x^{2}}\d x-p\int_{-2}^{2}\frac{\theta^{2}}{4\gamma^{4}}\frac{n^{2}}{p^{2}}x^{4}\frac{1}{2\pi}\sqrt{4-x^{2}}\d x+o\left(\frac{n^{2}}{p}\right)\\
 & =  p\log\gamma-\frac{\theta}{2\gamma^{2}}n-\frac{\theta^{2}}{2\gamma^{4}}\frac{n^{2}}{p}+o\left(\frac{n^{2}}{p}\right),
 \end{align*}

\begin{align*}
\Phi_{0}\left(f_{3}\right) & =  \frac{1}{2\pi}\int_{-\pi}^{\pi}\frac{p}{n}\log\left(\gamma+2\sqrt{\frac{n\theta}{p}}\cos t\right)\d t\\
 & =  \frac{p}{n}\log\gamma+\frac{1}{2\pi}\int_{-\pi}^{\pi}\frac{2\sqrt{\theta}}{\gamma}\sqrt{\frac{p}{n}}\cos t\d t-\frac{1}{2\pi}\int_{-\pi}^{\pi}\frac{2\theta}{\gamma^{2}}\left(\cos t\right)^{2}\d t\\
 & +  \frac{1}{2\pi}\int_{-\pi}^{\pi}\frac{8\theta\sqrt{\theta}}{3\gamma^{3}}\left(\cos t\right)^{3}\sqrt{\frac{n}{p}}\d t+o\left(\sqrt{\frac{n}{p}}\right)\\
 & =  \frac{p}{n}\log\gamma-\frac{\theta}{\gamma^{2}}+o\left(\sqrt{\frac{n}{p}}\right),
 \end{align*}

\begin{align*}
\Phi_{1}\left(f_{3}\right) & =  \frac{1}{2\pi}\int_{-\pi}^{\pi}\frac{p}{n}\log\left(\gamma+2\sqrt{\frac{n\theta}{p}}\cos t\right)\cos t\d t\\
 & =  \frac{1}{2\pi}\int_{-\pi}^{\pi}\frac{p}{n}\log\gamma\cos t\d t+\frac{1}{2\pi}\int_{-\pi}^{\pi}\frac{p}{n}\log\left(1+\frac{2\sqrt{\theta}}{\gamma}\sqrt{\frac{n}{p}}\cos t\right)\cos t\d t\\
 & =  \frac{1}{2\pi}\int_{-\pi}^{\pi}\frac{2\sqrt{\theta}}{\gamma}\sqrt{\frac{p}{n}}\left(\cos t\right)^{2}\d t-\frac{1}{2\pi}\int_{-\pi}^{\pi}\frac{2\theta}{\gamma^{2}}\left(\cos t\right)^{3}\d t\\
 & +  \frac{1}{2\pi}\int_{-\pi}^{\pi}\frac{8\theta\sqrt{\theta}}{3\gamma^{3}}\left(\cos t\right)^{4}\sqrt{\frac{n}{p}}\d t+o\left(\sqrt{\frac{n}{p}}\right)\\
 & =  \frac{\sqrt{\theta}}{\gamma}\sqrt{\frac{p}{n}}+\frac{\theta\sqrt{\theta}}{\gamma^{3}}\sqrt{\frac{n}{p}}+o\left(\sqrt{\frac{n}{p}}\right),
 \end{align*}

\begin{align*}
\Phi_{k}\left(f_{3}\right) & =  \frac{1}{2\pi}\int_{-\pi}^{\pi}\frac{p}{n}\log\left(\gamma+2\sqrt{\frac{n\theta}{p}}\cos t\right)\cos kt\d t\\
 & =  \frac{1}{2\pi}\int_{-\pi}^{\pi}\frac{p}{n}\log\gamma\cos kt\d t+\frac{1}{2\pi}\int_{-\pi}^{\pi}\frac{p}{n}\log\left(1+\frac{2\sqrt{\theta}}{\gamma}\sqrt{\frac{n}{p}}\cos t\right)\cos kt\d t\\
 & =  \frac{1}{2\pi}\int_{-\pi}^{\pi}\frac{2\sqrt{\theta}}{\gamma}\sqrt{\frac{p}{n}}\cos t\cos kt\d t-\frac{1}{2\pi}\int_{-\pi}^{\pi}\frac{2\theta}{\gamma^{2}}\left(\cos t\right)^{2}\cos kt\d t\\
 & +  \frac{1}{2\pi}\int_{-\pi}^{\pi}\frac{8\theta\sqrt{\theta}}{3\gamma^{3}}\sqrt{\frac{n}{p}}\left(\cos t\right)^{3}\cos kt\d t+o\left(\sqrt{\frac{n}{p}}\right)\\
 & =  \begin{cases}
-\frac{\theta}{2\gamma^{2}}, & k=2,\\
\frac{\theta\sqrt{\theta}}{3\gamma^{3}}\sqrt{\frac{n}{p}}, & k=3,\\
0, & k\geq4,\end{cases}
\end{align*}

thus,

\begin{align*}
G_{n}\left(f_{3}\right) & =  \sum_{i=1}^{n}\frac{p}{n}\log\left(\gamma+\sqrt{\frac{n\theta}{p}}\widetilde{\lambda}_{i}\right)-n\int_{-2}^{2}\frac{p}{n}\log\left(\gamma+\sqrt{\frac{n\theta}{p}}x\right)\frac{1}{2\pi}\sqrt{4-x^{2}}\d x-\sqrt{\frac{n^{3}}{p}}\Phi_{3}\left(f_{3}\right)\\
 & = \sum_{i=1}^{n}\frac{p}{n}\log\left(\gamma+\sqrt{\frac{n\theta}{p}}\widetilde{\lambda}_{i}\right)-\left(p\log\gamma-\frac{\theta}{2\gamma^{2}}n-\frac{\theta^{2}}{2\gamma^{4}}\frac{n^{2}}{p}\right)-\frac{\theta\sqrt{\theta}}{3\gamma^{3}}\frac{n^{2}}{p},
 \end{align*}

\begin{align*}
\E\left(Y\left(f_{3}\right)\right) & =  \frac{1}{4}\left(f_{3}\left(2\right)+f_{3}\left(-2\right)\right)-\frac{1}{2}\Phi_{0}\left(f_{3}\right)+\frac{\omega}{\theta}\left(\nu_{4}-3\right)\Phi_{2}\left(f_{3}\right)\\
 & =  \frac{1}{4}\left(\frac{p}{n}\log\left(\gamma+2\sqrt{\frac{n\theta}{p}}\right)+\frac{p}{n}\log\left(\gamma-2\sqrt{\frac{n\theta}{p}}\right)\right)\\
 & -  \frac{1}{2}\left(\frac{p}{n}\log\gamma-\frac{\theta}{\gamma^{2}}\right)-\frac{\omega}{\theta}\left(\nu_{4}-3\right)\frac{\theta}{2\gamma^{2}}\\
 & =  -\frac{\theta}{2\gamma^{2}}-\frac{\omega}{2\gamma^{2}}\left(\nu_{4}-3\right),
 \end{align*}

\begin{align*}
Var\left(Y\left(f_{3}\right)\right) & =  \frac{\omega}{\theta}\left(\nu_{4}-3\right)\Phi_{1}^{2}\left(f_{3}\right)+2\sum_{k=1}^{\infty}k\Phi_{k}^{2}\left(f_{3}\right)\\
 & =  \frac{\omega}{\theta}\left(\nu_{4}-3\right)\left(\frac{\sqrt{\theta}}{\gamma}\sqrt{\frac{p}{n}}+\frac{\theta\sqrt{\theta}}{\gamma^{3}}\sqrt{\frac{n}{p}}\right)^{2}\\
 & +  2\left(\frac{\sqrt{\theta}}{\gamma}\sqrt{\frac{p}{n}}+\frac{\theta\sqrt{\theta}}{\gamma^{3}}\sqrt{\frac{n}{p}}\right)^{2}+4\left(\frac{\theta^{2}}{4\gamma^{4}}\right)+6\frac{\theta^{3}}{9\gamma^{6}}\frac{n}{p}\\
 & =  \left(\frac{\omega}{\theta}\left(\nu_{4}-3\right)+2\right)\frac{\theta}{\gamma^{2}}\frac{p}{n}+\left(\frac{2\omega}{\theta}\left(\nu_{4}-3\right)+5\right)\frac{\theta^{2}}{\gamma^{4}}+\left(\frac{\omega}{\theta}\left(\nu_{4}-3\right)+\frac{8}{3}\right)\frac{\theta^{3}}{\gamma^{6}}\frac{n}{p}.
 \end{align*}

Consider function $f_{2}\left(x\right)=x$, from lemma \ref{Powerlem1}, we have

\[
\Phi_{k}\left(f_{2}\right)=\frac{1}{2\pi}\int_{-\pi}^{\pi}f_{2}\left(2\cos\left(\theta\right)\right)\cos\left(k\theta\right)\d\theta=\begin{cases}
0, & k=0\\
1, & k=1\\
0, & k\geq2\end{cases},\]

therefore the covariance between $Y\left(f_{2}\right)$ and $Y\left(f_{3}\right)$
is

\begin{align*}
Cov\left(Y\left(f_{2}\right),Y\left(f_{3}\right)\right) & =  \frac{\omega}{\theta}\left(\nu_{4}-3\right)\Phi_{1}\left(f_{2}\right)\Phi_{1}\left(f_{3}\right)+2\sum_{k=1}^{\infty}k\Phi_{k}\left(f_{2}\right)\Phi_{k}\left(f_{3}\right)\\
 & =  \left(\frac{\omega}{\theta}\left(\nu_{4}-3\right)+2\right)\left(\frac{\sqrt{\theta}}{\gamma}\sqrt{\frac{p}{n}}+\frac{\theta\sqrt{\theta}}{\gamma^{3}}\sqrt{\frac{n}{p}}\right),
 \end{align*}

consequently the result follows.
\end{proof}

\end{document}